%% file: tocs.tex
\documentclass[english]{article}
\usepackage{lmodern}

\usepackage[T1]{fontenc}
\usepackage[latin9]{inputenc}
\usepackage{color}
\usepackage{babel}
\usepackage{prettyref}
\usepackage{amsthm}
\usepackage{amsmath}
\usepackage{amssymb}
\usepackage{graphicx}
\usepackage[numbers]{natbib}
\usepackage[unicode=true,pdfusetitle,
 bookmarks=true,bookmarksnumbered=false,bookmarksopen=false,
 breaklinks=false,pdfborder={0 0 0},backref=page,colorlinks=true]
 {hyperref}
\hypersetup{
 linkcolor=blue, urlcolor=blue, citecolor=blue,  pdfpagemode=UseNone}

\makeatletter

\newcommand{\noun}[1]{\textsc{#1}}

\usepackage{enumitem}		% customizable list environments
\newcommand{\lyxaddress}[1]{
\par {\raggedright #1
\vspace{1.4em}
\noindent\par}
}
\theoremstyle{plain}
\newtheorem{thm}{\protect\theoremname}
  \theoremstyle{definition}
  \newtheorem{defn}[thm]{\protect\definitionname}
  \theoremstyle{plain}
  \newtheorem{lem}[thm]{\protect\lemmaname}
  \theoremstyle{plain}
  \newtheorem{prop}[thm]{\protect\propositionname}
  \theoremstyle{remark}
  \newtheorem{rem}[thm]{\protect\remarkname}

\@ifundefined{date}{}{\date{}}
\newrefformat{pro}{Proposition~\ref{#1}}
\newrefformat{cor}{Corollary~\ref{#1}}
\newrefformat{def}{Definition~\ref{#1}}
\newrefformat{sub}{Section~\ref{#1}}
\newrefformat{exa}{Example~\ref{#1}}
\newrefformat{rem}{Remark~\ref{#1}}

\usepackage{url}

\setlist{noitemsep}
\setlist[1]{leftmargin=1.7em}
\setlist[2]{leftmargin=1.7em}

\usepackage[font=small, width=.9\textwidth]{caption}

\usepackage{titlesec}
\titleformat{\section}{\large\bfseries}{\thesection}{1em}{}
\titleformat{\subsection}{\normalsize\bfseries}{\thesubsection}{1em}{}

\usepackage[table]{xcolor}
\RequirePackage{booktabs}
\RequirePackage{tikz}
\usetikzlibrary{arrows,automata,positioning,shadows,patterns,decorations.pathreplacing}

\newcommand{\cloneFont}[1]{\mathsf{#1}}
\newcommand{\CloneBF}{\protect\ensuremath{\cloneFont{BF}}}
\newcommand{\CloneM}{\protect\ensuremath{\cloneFont{M}}}
\newcommand{\CloneL}{\protect\ensuremath{\cloneFont{L}}}
\newcommand{\CloneR}{\protect\ensuremath{\cloneFont{R}}}
\newcommand{\CloneD}{\protect\ensuremath{\cloneFont{D}}}
\newcommand{\CloneN}{\protect\ensuremath{\cloneFont{N}}}
\newcommand{\CloneS}{\protect\ensuremath{\cloneFont{S}}}
\newcommand{\CloneV}{\protect\ensuremath{\cloneFont{V}}}
\newcommand{\CloneE}{\protect\ensuremath{\cloneFont{E}}}
\newcommand{\CloneI}{\protect\ensuremath{\cloneFont{I}}}

\definecolor{g1}{gray}{0.7}
\definecolor{g3}{gray}{0.87}

\colorlet{WSAT}{g1}
\colorlet{nWSAT}{g3}

\tikzset{
WSAT/.style={fill=WSAT},
nWSAT/.style={fill=nWSAT},
}

\makeatother

  \providecommand{\definitionname}{Definition}
  \providecommand{\lemmaname}{Lemma}
  \providecommand{\propositionname}{Proposition}
  \providecommand{\remarkname}{Remark}
\providecommand{\theoremname}{Theorem}

\begin{document}

\title{\vspace{-6.5ex}The Connectivity of Boolean Satisfiability: Dichotomies
for Formulas and Circuits}

\author{\noun{\normalsize{}Konrad W. Schwerdtfeger}}

\maketitle

\lyxaddress{\begin{center}
{\small{}\vspace{-5ex}Institut für Theoretische Informatik, Leibniz
Universität Hannover,}\\
{\small{}Appelstr. 4, 30167 Hannover, Germany}\\
{\small{}\path|k.w.s@gmx.net|}
\par\end{center}}
\begin{abstract}
For Boolean satisfiability problems, the structure of the solution
space is characterized by the solution graph, where the vertices are
the solutions, and two solutions are connected iff they differ in
exactly one variable. In 2006, Gopalan et al.\ studied connectivity
properties of the solution graph and related complexity issues for
CSPs \citep{gop}, motivated mainly by research on satisfiability
algorithms and the satisfiability threshold. They proved dichotomies
for the diameter of connected components and for the complexity of
the $st$-connectivity question, and conjectured a trichotomy for
the connectivity question. Recently, we were able to establish the
trichotomy \citep{csp}.

Here, we consider connectivity issues of satisfiability problems defined
by Boolean circuits and propositional formulas that use gates, resp.
connectives, from a fixed set of Boolean functions. We obtain dichotomies
for the diameter and the two connectivity problems: On one side, the
diameter is linear in the number of variables, and both problems are
in P, while on the other side, the diameter can be exponential, and
the problems are PSPACE-complete. For partially quantified formulas,
we show an analogous dichotomy.\\

\end{abstract}

\section{Introduction}

The Boolean satisfiability problem (SAT), as well as many related
questions like equivalence, counting, enumeration, and numerous versions
of optimization, are of great importance in both theory and applications
of computer science. In this article, we focus on the solution-space
structure: We consider the \emph{solution graph}, where the vertices
are the solutions, and two solutions are connected iff they differ
in exactly one variable. For this implicitly defined graph, we then
study the connectivity and $st$-connectivity problems, and the diameter
of connected components. The figures below give an impression of how
solution graphs may look like.

\begin{figure}[!h]
\begin{centering}
\includegraphics[scale=0.64]{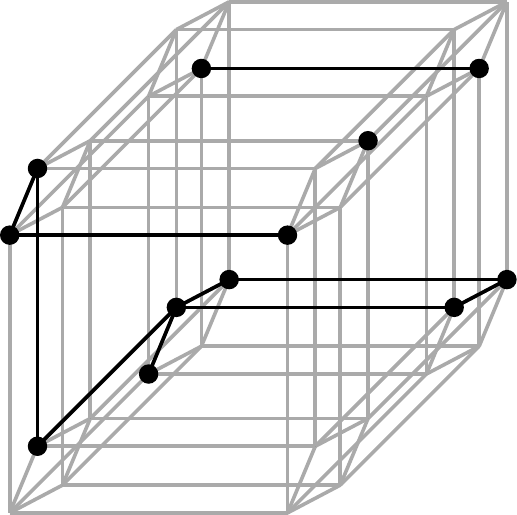}$\qquad$\includegraphics[scale=0.34]{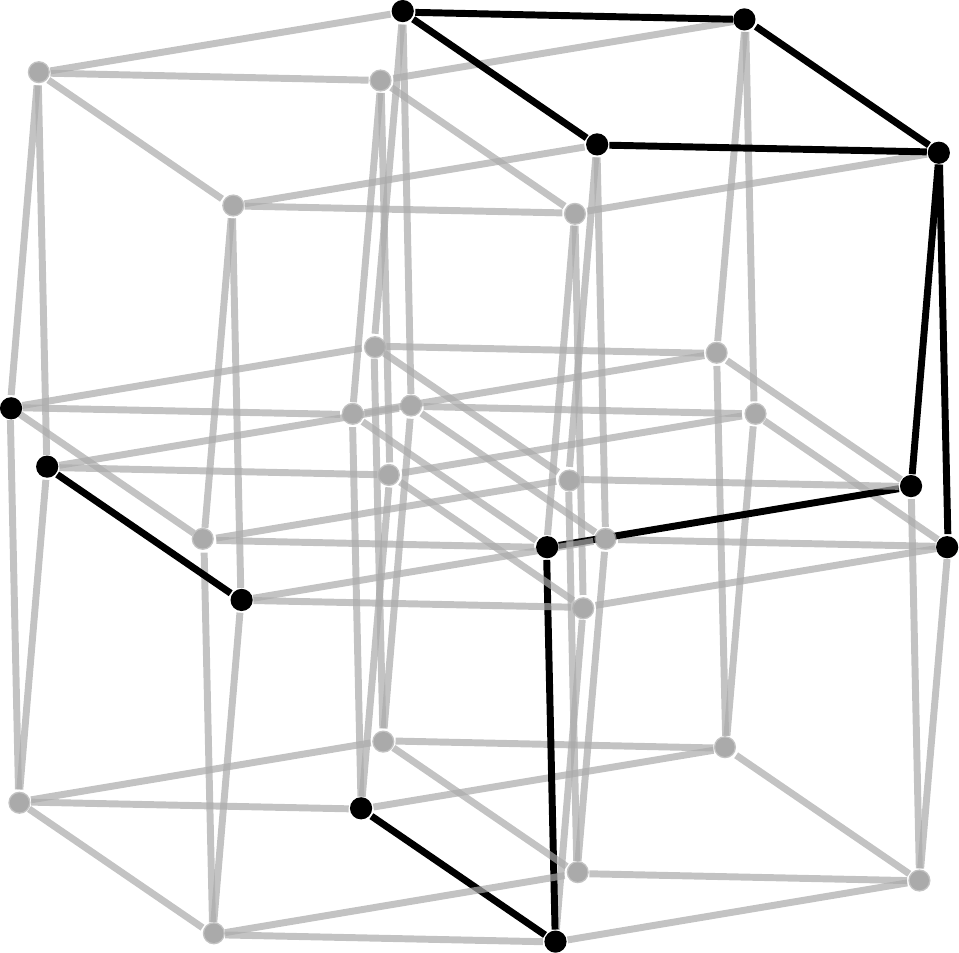}$\qquad$\includegraphics[scale=0.36]{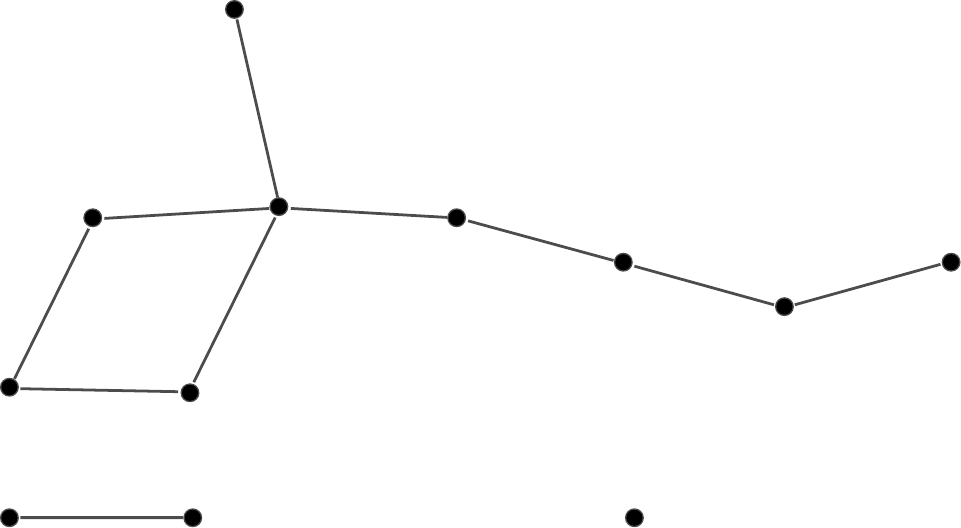}
\par\end{centering}

\protect\caption{Depictions of the subgraph of the 5-dimensional hypercube graph induced
by a typical random Boolean relation with 12 elements. Left: highlighted
on a orthographic hypercube projection. Center: highlighted on a ``Spectral
Embedding'' of the hypercube graph by \noun{Mathematica}. Right:
the sole subgraph, arranged by \noun{Mathematica}}
\end{figure}

\begin{figure}[!h]
\begin{centering}
\includegraphics[scale=0.3]{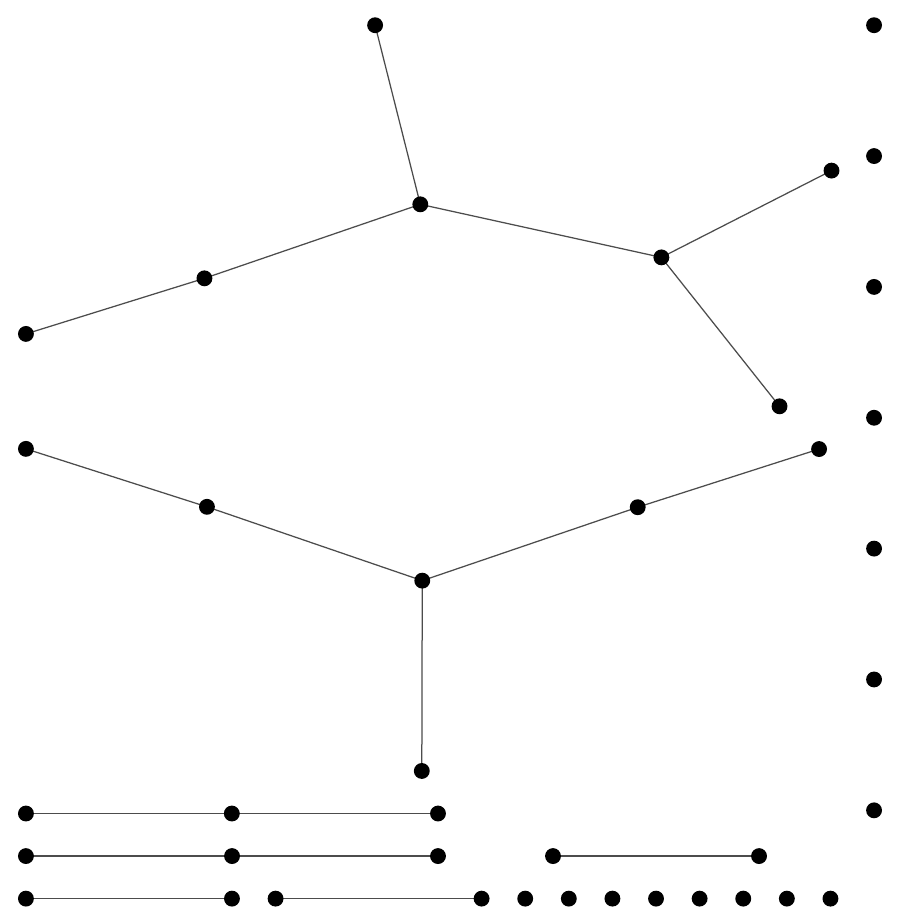}$\qquad$\includegraphics[scale=0.36]{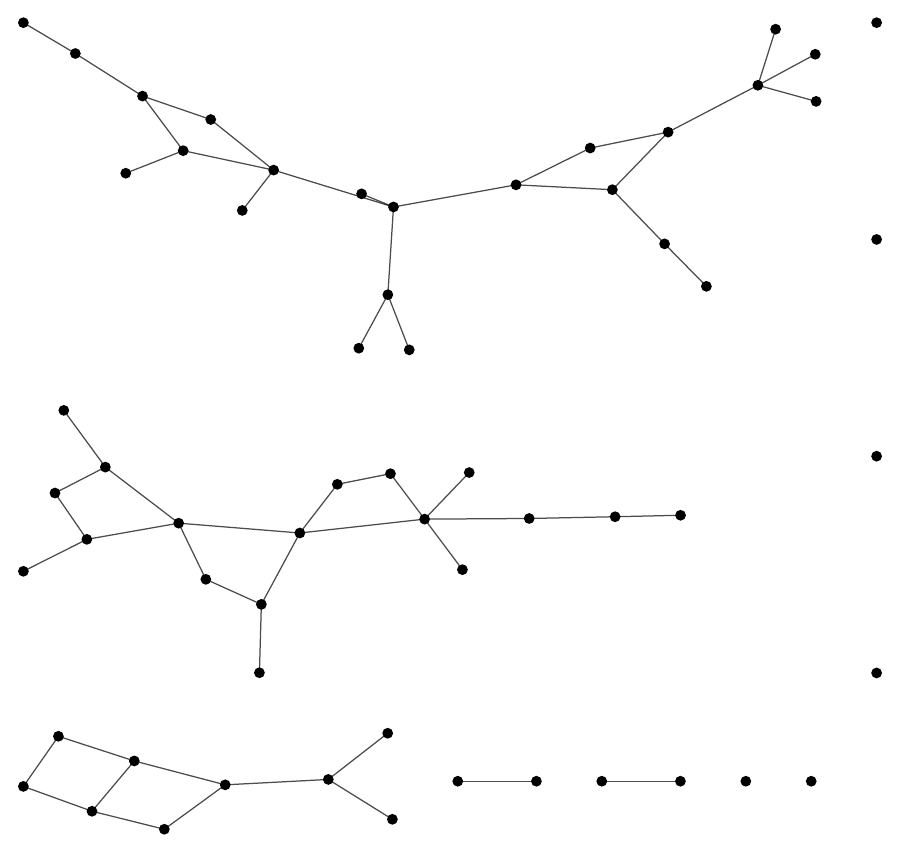}$\qquad$\includegraphics[scale=0.46]{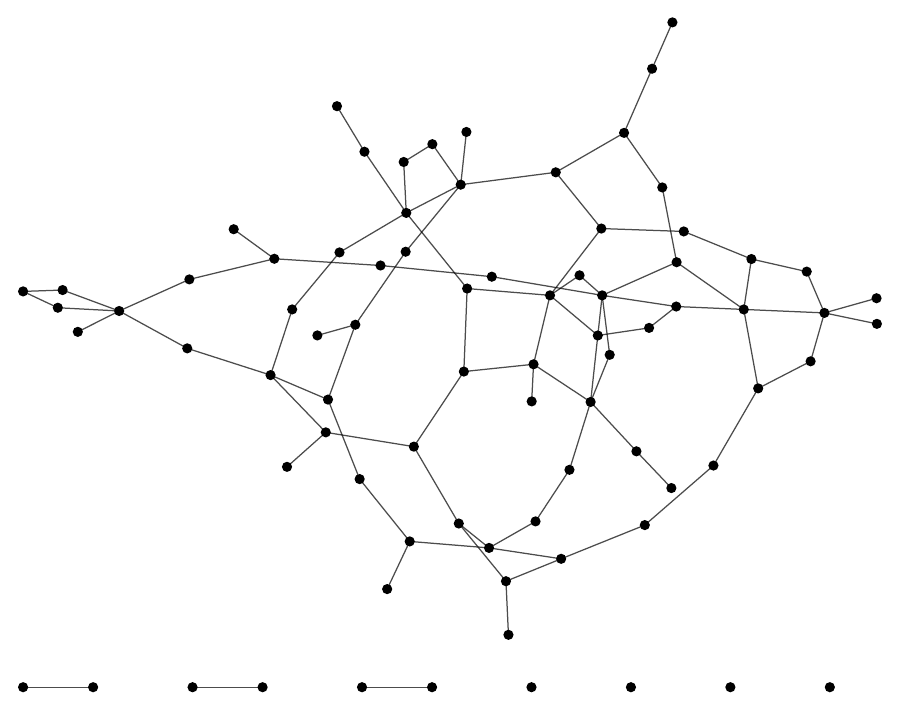}
\par\end{centering}

\protect\caption{Subgraphs of the 8-dimensional hypercube graph induced by typical
random relations with 40, 60 and 80 elements, arranged by \noun{Mathematica.}}
\end{figure}

While the standard satisfiability problem is defined for propositional
formulas, which can be seen as one special form of descriptions for
Boolean relations, satisfiability and related problems have also been
considered for many alternative descriptions, e.g. Boolean constraint
satisfactions problems (\emph{CSP}s), Boolean circuits, binary decision
diagrams, and Boolean neural networks. For the usual formulas with
the connectives $\wedge$, $\vee$ and $\neg$, there are several
common variants. A special form are formulas in conjunctive normal
form (CNF-formulas). A generalization of CNF-formulas are CNF($\mathcal{S}$)-formulas,
which are conjunctions of constraints on the variables taken from
a finite template set $\mathcal{S}$.

Here we consider another type of generalization: Arbitrarily nested
formulas built with connectives from some finite set of Boolean functions
$B$ (where the arity may be greater than two), known as \emph{$B$-formulas}.
Also we study \emph{$B$-circuits}, where analogously the allowed
gates implement the functions from $B$. As a further extension we
consider partially quantified \emph{$B$-}formulas.

A direct application of $st$-connectivity in solution graphs are
\emph{reconfiguration problems}, that arise when we wish to find a
step-by-step transformation between two feasible solutions of a problem,
such that all intermediate results are also feasible. Recently, the
reconfiguration versions of many problems such as \noun{Independent-Set},
\noun{Vertex-Cover}, \noun{Set-Cover} \noun{Graph-$k$-Coloring},
\noun{Shortest-Path} have been studied, and complexity results obtained
(see e.g. \citep{recon,short}). Also of relevance are the connectivity
properties to the problem of \emph{structure identification}, where
one is given a relation explicitly and seeks a short representation
of some kind (see e.g. \citep{creignou2008structure}); this problem
is important especially for learning in artificial intelligence.

A better understanding of the solution space structure also promises
advancement of SAT algorithms: It has been discovered that the solution
space connectivity is strongly correlated to the performance of standard
satisfiability algorithms like WalkSAT and DPLL on random instances:
As one approaches the\emph{ satisfiability threshold} (the ratio of
constraints to variables at which random $k$-CNF-formulas become
unsatisfiable for $k\geq3$) from below, the solution space (with
the connectivity defined as above) fractures, and the performance
of the algorithms deteriorates \citep{mezard2005clustering,maneva2007new}.
These insights mainly came from statistical physics, and lead to the
development of the \emph{survey propagation algorithm}, which has
much better performance on random instances \citep{maneva2007new}.

While current SAT solvers normally accept only CNF-formulas as input,
one of the most important applications of satisfiability testing is
verification and optimization in Electronic Design Automation (EDA),
where the instances derive mostly from digital circuit descriptions
\citep{csat2}. Though many such instances can easily be encoded in
CNF, the original structural information, such as signal ordering,
gate orientation and logic paths, is lost, or at least obscured. Since
exactly this information can be very helpful for solving these instances,
considerable effort has been made recently to develop satisfiability
solvers that work with the circuit description directly \citep{csat2},
which have far superior performance in EDA applications, or to restore
the circuit structure from CNF \citep{extracting}. This is a major
motivation for our study.

Our perspective is mainly from complexity theory: We classify $B$-formulas
and $B$-circuits by the worst-case complexity of the connectivity
problems, analogously to Schaefer's dichotomy theorem for satisfiability
of CSPs from 1978 \citep{Schaefer:1978:CSP:800133.804350}, Lewis'
dichotomy for satisfiability of $B$-formulas from 1979 \citep{lewis1979},
and Gopalan et al.'s classification for the connectivity problems
of CSPs from 2006 \citep{gop}. Along the way, we will examine structural
properties of the solution graph like its maximal diameter, and devise
efficient algorithms for solving the connectivity problems.

We begin with a formal definition of some central concepts.
\begin{defn}
An $n$-ary\emph{ Boolean relation} is a subset of $\{0,1\}^{n}$
($n\ge1$). If $\phi$ is some description of an $n$-ary Boolean
relation $R$, e.g. a propositional formula (where the variables are
taken in lexicographic order), the \emph{solution graph} $G(\phi)$
of $\phi$ is the subgraph of the $n$-dimensional hypercube graph
induced by the vectors in $R$, i.e., the vertices of $G(\phi$) are
the vectors in $R$, and there is an edge between two vectors precisely
if they differ in exactly one position. 

We use $\boldsymbol{a},\boldsymbol{b},\ldots$ to denote vectors of
Boolean values and $\boldsymbol{x},\boldsymbol{y},\ldots$ to denote
vectors of variables, $\boldsymbol{a}=(a_{1},a_{2},\ldots)$ and $\boldsymbol{x}=(x_{1},x_{2},\ldots)$.

The \emph{Hamming weight }$|\boldsymbol{a}|$ of a Boolean vector
\textbf{$\boldsymbol{a}$} is the number of 1's in \textbf{$\boldsymbol{a}$}.
For two vectors \textbf{$\boldsymbol{a}$} and $\boldsymbol{b}$,
the \emph{Hamming distance }$|\boldsymbol{a}-\boldsymbol{b}|$ is
is the number of positions in which they differ\emph{.}

If \textbf{$\boldsymbol{a}$} and $\boldsymbol{b}$ are solutions
of $\phi$ and lie in the same connected component of $G(\phi)$,
we write $d_{\phi}(\boldsymbol{a},\boldsymbol{b})$ to denote the
shortest-path distance between \textbf{$\boldsymbol{a}$} and $\boldsymbol{b}$.

The \emph{diameter} \emph{of a connected component} is the maximal
shortest-path distance between any two vectors in that component.
The \emph{diameter of }$G(\phi)$\emph{ }is the maximal diameter of
any of its connected components.
\end{defn}

\section{Connectivity of CNF-Formulas}

Research has focused on the structure of the solution space only quite
recently: One of the earliest studies on solution-space connectivity
was done for CNF($\mathcal{S}$)-formulas with constants (see the
definition below), begun in 2006 by Gopalan et al. (\citep{Gopalan:2006},
\citep{Makino:2007:BCP:1768142.1768162}, \citep{gop}, \citep{csp}). 

In our proofs for $B$-formulas and $B$-circuits, we will use Gopalan
et al.'s results for 3-CNF-formulas, so we have to introduce some
related terminology.
\begin{defn}
\label{def:cnf}A \emph{CNF-formula} is a Boolean formula of the form
$C_{1}\wedge\cdots\wedge C_{m}$ ($1\leq m<\infty$), where each $C_{i}$
is a \emph{clause}, that is, a finite disjunction of \emph{literals}
(variables or negated variables). A\emph{ $k$-CNF-formula} ($k\geq1$)
is a CNF-formula where each $C_{i}$ has at most $k$ literals.

For a finite set of Boolean relations $\mathcal{S}$, a\emph{ CNF($\mathcal{S}$)-formula}
(with constants) over a set of variables $V$ is a finite conjunction
$C_{1}\wedge\cdots\wedge C_{m}$, where each $C_{i}$ is a \emph{constraint
application} (\emph{constraint }for short), i.e., an expression of
the form $R(\xi_{1},\ldots,\xi_{k})$, with a $k$-ary relation $R\in\mathcal{S}$,
and each $\xi_{j}$ is a variable in $V$ or one of the constants
0, 1.

A \emph{$k$-clause} is a disjunction of $k$ variables or negated
variables. For $0\leq i\leq k$, let $D_{i}$ be the set of all satisfying
truth assignments of the $k$-clause whose first $i$ literals are
negated, and let $S_{k}=\{D_{0},\ldots,D_{k}\}$. Thus, CNF($S_{k}$)
is the collection of $k$-CNF-formulas.
\end{defn}
Gopalan et al.\ studied the following two decision problems for CNF($\mathcal{S}$)-formulas:
\begin{itemize}
\item the \emph{connectivity problem }\emph{\noun{Conn}}\emph{($\mathcal{S}$)}:
given a CNF($\mathcal{S}$)-formula $\phi$, is $G(\phi)$ connected?
(if $\phi$ is unsatisfiable, then $G(\phi)$ is considered connected)
\item the \emph{$st$-connectivity problem }\emph{\noun{st-Conn}}\emph{($\mathcal{S}$)}:
given a CNF($\mathcal{S}$)-formula $\phi$ and two solutions $\boldsymbol{s}$
and $\boldsymbol{t}$, is there a path from $\boldsymbol{s}$ to $\boldsymbol{t}$
in $G(\phi)$?\end{itemize}
\begin{lem}
\label{lem:-gopa} \emph{\citep[Lemm 3.6]{gop} }\noun{st-Conn($S_{3}$)}
and \noun{Conn($S_{3}$)} are $\mathrm{PSPACE}$-complete.
\end{lem}
Showing that the problems are in PSPACE is straightforward: Given
a CNF($S_{3}$)-formula $\phi$ and two solutions $\boldsymbol{s}$
and $\boldsymbol{t}$, we can guess a path of length at most $2^{n}$
between them and verify that each vertex along the path is indeed
a solution. Hence \noun{st-Conn}($S_{3}$) is in $\mathrm{NPSPACE}$,
which equals PSPACE by Savitch\textquoteright s theorem. For \noun{Conn}($S_{3}$),
by reusing space we can check for all pairs of vectors whether they
are satisfying, and, if they both are, whether they are connected
in $G(\phi)$.

The hardness-proof is quite intricate: it consists of a direct reduction
from the computation of a space-bounded Turing machine $M$. The input-string
$w$ of $M$ is mapped to a CNF($S_{3}$)-formula $\phi$ and two
satisfying assignments $\boldsymbol{s}$ and $\boldsymbol{t}$, corresponding
to the initial and accepting configuration of a Turing machine $M'$
constructed from $M$ and $w$, s.t. $\boldsymbol{s}$ and $\boldsymbol{t}$
are connected in $G(\phi)$ iff $M$ accepts $w$. Further, all satisfying
assignments of $\phi$ are connected to either $\boldsymbol{s}$ or
$\boldsymbol{t}$, so that $G(\phi)$ is connected iff $M$ accepts
$w$.
\begin{lem}
\label{lem:gop}\emph{\citep[Lemm 3.7]{gop} }For $n\geq2$ , there
is an $n$-ary Boolean function $f$ with $f(1,\ldots,1)=1$ and a
diameter of at least $2^{\left\lfloor \frac{n}{2}\right\rfloor }$.
\end{lem}
The proof of this lemma is by direct construction of such a formula.

\section{\label{sec:Cir}Circuits, Formulas, and Post's Lattice}

An $n$-ary\emph{ Boolean function} is a function $f:\{0,1\}^{n}\rightarrow\{0,1\}$.
Let $B$ be a finite set of Boolean functions.

A \emph{$B$-circuit} $\mathcal{C}$ with input variables $x_{1},\ldots,x_{n}$
is a directed acyclic graph, augmented as follows: Each node (here
also called \emph{gate}) with indegree 0 is labeled with an $x_{i}$
or a 0-ary function from $B$, each node with indegree $k>0$ is labeled
with a $k$-ary function from $B$. The edges (here also called \emph{wires})
pointing into a gate are ordered. One node is designated the output
gate. Given values $a_{1},\ldots,a_{n}\in\{0,1\}$ to $x_{1},\ldots,x_{n}$,
$\mathcal{C}$ computes an $n$-ary function $f_{\mathcal{C}}$ as
follows: A gate $v$ labeled with a variable $x_{i}$ returns $a_{i}$,
a gate $v$ labeled with a function $f$ computes the value $f(b_{1},\ldots,b_{k})$,
where $b_{1},\ldots,b_{k}$ are the values computed by the predecessor
gates of $v$, ordered according to the order of the wires. For a
more formal definition see \citep{Vollmer:1999:ICC:520668}.

A \emph{$B$-formula} is defined inductively: A variable $x$ is a
$B$-formula. If $\phi_{1},\ldots,\phi_{m}$ are $B$-formulas, and
$f$ is an $n$-ary function from $B$, then $f(\phi_{1},\ldots,\phi_{n})$
is a $B$-formula. In turn, any $B$-formula defines a Boolean function
in the obvious way, and we will identify $B$-formulas and the function
they define.

It is easy to see that the functions computable by a $B$-circuit,
as well as the functions definable by a $B$-formula, are exactly
those that can be obtained from $B$ by \emph{superposition}, together
with all projections \citep{bloc}. By superposition, we mean substitution
(that is, composition of functions), permutation and identification
of variables, and introduction of \emph{fictive variables} (variables
on which the value of the function does not depend). This class of
functions is denoted by $[B]$. $B$ is \emph{closed} (or said to
be a \emph{clone}) if $[B]=B$. A \emph{base} of a clone $F$ is any
set $B$ with $[B]=F$.

Already in the early 1920s, Emil Post extensively studied Boolean
functions \citep{post1941two}. He identified all clones, found a
finite base for each of them, and detected their inclusion structure:
The clones form a lattice, called \emph{Post's lattice}, depicted
in Figure \ref{fig:Post's-lattice}. 

The following clones are defined by properties of the functions they
contain, all other ones are intersections of these. Let $f$ be an
$n$-ary Boolean function.
\begin{itemize}
\item $\mathsf{BF}$ is the class of all Boolean functions.
\item $\mathsf{R}_{0}$ ($\mathsf{R}_{1}$) is the class of all 0-reproducing
(1-reproducing) functions,\\
$f$\emph{ }is\emph{ $c$-reproducing}, if $f(c,\ldots,c)=c$, where
$c\in\{0,1\}$.
\item $\mathsf{M}$ is is the class of all monotone functions,\emph{}\\
$f$\emph{ }is\emph{ monotone}, if $a_{1}\leq b_{1},\ldots,a_{n}\leq b_{n}$
implies $f(a_{1},\ldots,a_{n})\leq f(b_{1},\ldots,b_{n})$.
\item $\mathsf{D}$ is the class of all self-dual functions,\emph{}\\
$f$\emph{ }is\emph{ self-dual}, if $f(x_{1},\ldots,x_{n})=\overline{f(\overline{x_{1}},\ldots,\overline{x_{n}})}$.
\item $\mathsf{L}$ is the class of all affine (on \emph{linear}) functions,\emph{}\\
$f$\emph{ }is\emph{ affine}, if $f(x_{1},\ldots,x_{n})=x_{i_{1}}\oplus\cdots\oplus x_{i_{m}}\oplus c$
with $i_{1},\ldots,i_{m}\in\{1,\ldots,n\}$ and $c\in\{0,1\}$.
\item $\mathsf{S}_{0}$ ($\mathsf{S}_{1}$) is the class of all 0-separating
(1-separating) functions,\emph{}\\
$f$\emph{ }is\emph{ $c$-separating}, if there exists an $i\in\{1,\ldots,n\}$
s.t. $a_{i}=c$ for all $\boldsymbol{a}\in f^{-1}(c)$, where $c\in\{0,1\}$.
\item $\mathsf{S}_{0}^{m}$ ($\mathsf{S}_{1}^{m}$) is the class of all
functions that are 0-separating (1-separating) of degree $m$,\\
$f$\emph{ }is\emph{ $c$-separating of degree $m$}, if for all $U\subseteq f^{-1}(c)$
of size $|U|=m$ there exists an $i\in\{1,\ldots,n\}$ s.t. $a_{i}=c$
for all $\boldsymbol{a}\in U$ ($c\in\{0,1\}$, $m\geq2$).
\end{itemize}
The definitions and bases of all classes are given in Table \ref{tab:List-of-all}.
For an introduction to Post's lattice and further references see e.g.
\citep{bloc}.

\begin{figure}[p]
\begin{tikzpicture}[clone/.style={circle,inner sep=0,minimum width=6mm,fill=white,font=\small},x=0.95cm,y=0.75cm,
edge/.style={draw=black,thin,-},
clonefill/.style={draw=black,circle,minimum width=6.1mm,inner sep=0,font=\scriptsize},
clonefil/.style={draw=black,very thick,circle,minimum width=6.1mm,inner sep=0,font=\scriptsize},scale=0.87, transform shape]

\input{lattice}

	% Clone Fillings
	% BF and c-Reproducing Functions
	\node[clonefil,fill=WSAT] at	 (BF) 		{$\CloneBF$};
	\node[clonefill,fill=WSAT] at  (R1)  	{$\CloneR_1$};
	\node[clonefil,fill=WSAT] at  (R0) 		{$\CloneR_0$};
	\node[clonefill,fill=WSAT] at  (R2)  	{$\CloneR_2$};
	
	% Montone Functions
	\node[clonefill,fill=white] at  (M)  		{$\CloneM$};
	\node[clonefill,fill=white] at  (M1)  	{$\CloneM_1$};
	\node[clonefill,fill=white] at  (M0) 		{$\CloneM_0$};
	\node[clonefill,fill=white] at  (M2)  	{$\CloneM_2$};

	% 1-separating 1
	\node[clonefil,fill=WSAT] at  (S21) 	{$\CloneS_{1}^2$};
	\node[clonefil,fill=WSAT] at  (S31)		{$\CloneS_{1}^3$};
	\node[clonefil,fill=WSAT] at  (Sn1)		{$\CloneS_{1}^n$};
	\node[clonefil,fill=WSAT] at  (S1)		{$\CloneS_{1}$};

	% 1-separating 2
	\node[clonefill,fill=WSAT] at  (S212) 	{$\CloneS_{12}^2$};
	\node[clonefill,fill=WSAT] at  (S312) 	{$\CloneS_{12}^3$};
	\node[clonefill,fill=WSAT] at  (Sn12) 	{$\CloneS_{12}^n$};
	\node[clonefill,fill=WSAT] at  (S12)	  	{$\CloneS_{12}$};

	% 1-separating 3
	\node[clonefill,fill=white] at  (S211) 	{$\CloneS_{11}^2$};
	\node[clonefill,fill=white] at  (S311) 	{$\CloneS_{11}^3$};
	\node[clonefill,fill=white] at  (Sn11) 	{$\CloneS_{11}^n$};
	\node[clonefill,fill=white] at  (S11)	  	{$\CloneS_{11}$};

	% 1-separating 4
	\node[clonefill,fill=white] at  (S210) 	{$\CloneS_{10}^2$};
	\node[clonefill,fill=white] at  (S310) 	{$\CloneS_{10}^3$};
	\node[clonefill,fill=white] at  (Sn10) 	{$\CloneS_{10}^n$};
	\node[clonefill,fill=white] at  (S10)	  	{$\CloneS_{10}$};

	% 0-separating 1
	\node[clonefill,fill=WSAT] at  (S20) 	{$\CloneS_{0}^2$};
	\node[clonefill,fill=WSAT] at  (S30) 	{$\CloneS_{0}^3$};
	\node[clonefill,fill=WSAT] at  (Sn0) 	{$\CloneS_{0}^n$};
	\node[clonefill,fill=nWSAT] at  (S0)	  	{$\CloneS_{0}$};

	% 0-separating 2
	\node[clonefill,fill=WSAT] at  (S202) 	{$\CloneS_{02}^2$};
	\node[clonefill,fill=WSAT] at  (S302) 	{$\CloneS_{02}^3$};
	\node[clonefill,fill=WSAT] at  (Sn02) 	{$\CloneS_{02}^n$};
	\node[clonefill,fill=nWSAT] at  (S02)	  	{$\CloneS_{02}$};

	% 0-separating 3
	\node[clonefill,fill=white] at  (S201) 	{$\CloneS_{01}^2$};
	\node[clonefill,fill=white] at  (S301) 	{$\CloneS_{01}^3$};
	\node[clonefill,fill=white] at  (Sn01) 	{$\CloneS_{01}^n$};
	\node[clonefill,fill=white] at  (S01)	  	{$\CloneS_{01}$};
	
	% 0-separating 4
	\node[clonefill,fill=white] at  (S200) 	{$\CloneS_{00}^2$};
	\node[clonefill,fill=white] at  (S300) 	{$\CloneS_{00}^3$};
	\node[clonefill,fill=white] at  (Sn00) 	{$\CloneS_{00}^n$};
	\node[clonefill,fill=white] at  (S00)	  	{$\CloneS_{00}$};
	
	% Selfdual Functions
	\node[clonefill,fill=WSAT] at  (D)  		{$\CloneD$};
	\node[clonefill,fill=WSAT] at  (D1)  	{$\CloneD_1$};
	\node[clonefill,fill=white] at  (D2)  	{$\CloneD_2$};
	
	% Conjunctions
	\node[clonefill,fill=white] at  (E)  		{$\CloneE$};
	\node[clonefill,fill=white] at  (E1)  	{$\CloneE_1$};
	\node[clonefill,fill=white] at  (E0) 		{$\CloneE_0$};
	\node[clonefill,fill=white] at  (E2)  	{$\CloneE_2$};
	
	% Disjunctions
	\node[clonefill,fill=white] at  (V)  		{$\CloneV$};
	\node[clonefill,fill=white] at  (V0) 		{$\CloneV_0$};
	\node[clonefill,fill=white] at  (V1)  	{$\CloneV_1$};
	\node[clonefill,fill=white] at  (V2)  	{$\CloneV_2$};
	
	% Affine Functions
	\node[clonefill,fill=white] at  (L)  		{$\CloneL$};
	\node[clonefill,fill=white] at  (L0)  	{$\CloneL_0$};
	\node[clonefill,fill=white] at  (L1)  	{$\CloneL_1$};
	\node[clonefill,fill=white] at  (L3) 		{$\CloneL_3$};
	\node[clonefill,fill=white] at  (L2)  	{$\CloneL_2$};
	
	% Negations
	\node[clonefill,fill=white] at  (N)  		{$\CloneN$};
	\node[clonefill,fill=white] at  (N2)  	{$\CloneN_2$};
	
	% Identities
	\node[clonefill,fill=white] at  (I)  		{$\CloneI$};
	\node[clonefill,fill=white] at  (I0)  	{$\CloneI_0$};
	\node[clonefill,fill=white] at  (I1)  	{$\CloneI_1$};
	\node[clonefill,fill=white] at  (I2)  	{$\CloneI_2$};
	
\end{tikzpicture}
    \bigskip

\protect\caption{\label{fig:Post's-lattice}Graphical representation of Post's lattice.\protect \\
The classes on the hard side of the dichotomy for the connectivity
problems and the diameter are shaded gray; the light gray shaded ones
are only on the hard side for formulas with quantifiers.\protect \\
For comparison, the classes for which SAT (without quantifiers) is
NP-complete are circled bold.}
\end{figure}

\begin{table}[p]

  \fontsize{8}{9}\selectfont
    \setlength{\tabcolsep}{6pt}
    \rowcolors{2}{gray!10}{}
    \centering
    \begin{tabular}{llll}
    \toprule
    Class  & Definition & Base\\
    \midrule
    $\CloneBF$    & All Boolean functions & $\{x \land y ,\neg x\}$ &\\
    $\CloneR_0$   & $\{f \in \CloneBF \mid f \mbox{ is 0-reproducing}\}$ & $\{ x \land y , x \oplus y \}$ &\\
    $\CloneR_1$   & $\{f \in \CloneBF \mid f \mbox{ is 1-reproducing}\}$ & $\{x \lor y, x \leftrightarrow y \}$ &\\
    $\CloneR_2$   & $\CloneR_0 \cap \CloneR_1$ & $\{x \lor y, x \land (y \leftrightarrow z) \}$ &\\
    $\CloneM$     & $\{f \in \CloneBF \mid f \mbox{ is monotone}\}$ & $\{x \land y, x\lor y,0,1\}$ &\\
    $\CloneM_0$   & $\CloneM \cap \CloneR_0$ & $\{x \land y,x\lor y,0\}$ &\\
    $\CloneM_1$   & $\CloneM \cap \CloneR_1$ & $\{x \land y,x\lor y,1\}$ &\\  
    $\CloneM_2$   & $\CloneM \cap \CloneR_2$ & $\{x \land y,x\lor y\}$ &\\
    $\CloneS_0$   & $\{f \in \CloneBF \mid f \mbox{ is 0-separating}\}$ & $\{x \rightarrow y\}$ &\\
    $\CloneS_0^n$ & $\{f \in \CloneBF \mid f \mbox{ is 0-separating of degree $n$}\}$ & $\{x \rightarrow y,\mbox{dual}({\mathrm{T}^{n+1}_n})\}$ &\\
    $\CloneS_1$   & $\{f \in \CloneBF \mid f \mbox{ is 1-separating}\}$ & $\{x \nrightarrow y\}$ &\\
    $\CloneS_1^n$ & $\{f \in \CloneBF \mid f \mbox{ is 1-separating of degree $n$}\}$ & $\{x \nrightarrow y,\mathrm{T}^{n+1}_n\}$ &\\
    $\CloneS_{02}^n$& $\CloneS_0^n \cap \CloneR_2$ & $\{x \lor (y \land \neg z), \mbox{dual}({\mathrm{T}^{n+1}_n})\}$ &\\
    $\CloneS_{02}$  & $\CloneS_0 \cap \CloneR_2$ & $\{x \lor (y \land \neg z)\}$ &\\
    $\CloneS_{01}^n$& $\CloneS_0^n \cap \CloneM$ & $\{\mbox{dual}({\mathrm{T}^{n+1}_n}),1\}$ &\\
    $\CloneS_{01}$  & $\CloneS_0 \cap \CloneM$ & $\{x \lor (y \land z),1\}$ &\\
    $\CloneS_{00}^n$& $\CloneS_0^n \cap \CloneR_2 \cap \CloneM$ 
    	&$\{x \lor (y \land z),\mbox{dual}({\mathrm{T}^{n+1}_n})\}$& \\
    $\CloneS_{00}$  & $\CloneS_0 \cap \CloneR_2 \cap \CloneM$ & $\{x \lor (y \land z)\}$ &\\
    $\CloneS_{12}^n$& $\CloneS_1^n \cap \CloneR_2$ & $\{x\land (y\lor \neg z),\mathrm{T}^{n+1}_n\}$ &\\
    $\CloneS_{12}$  & $\CloneS_1 \cap \CloneR_2$ & $\{x\land (y\lor \neg z)\}$ &\\
    $\CloneS_{11}^n$& $\CloneS_1^n \cap \CloneM$ & $\{\mathrm{T}^{n+1}_n,0\}$ &\\
    $\CloneS_{11}$  & $\CloneS_1 \cap \CloneM$ & $\{x\land (y\lor z),0\}$ &\\
    $\CloneS_{10}^n$& $\CloneS_1^n \cap \CloneR_2 \cap \CloneM$ 
    & $\{x \land (y \lor z),{\mathrm{T}^{n+1}_n}\}$ &\\
    $\CloneS_{10}$  & $\CloneS_1 \cap \CloneR_2 \cap \CloneM$ & $\{x\land (y\lor z)\}$ &\\
    
    $\CloneD$     & $\{f \in \CloneBF \mid f \mbox{ is self-dual}\}$ & $\{\mbox{maj}(x,\neg y,\neg z)\}$ &\\
    $\CloneD_1$   & $\CloneD \cap \CloneR_2$ & $\{\mbox{maj}(x, y,\neg z)\}$ &\\
    $\CloneD_2$   & $\CloneD \cap \CloneM$ & $\{\mbox{maj}(x,y,z)\}$ &\\  
    
    $\CloneL$     & $\{f \in \CloneBF \mid f \mbox{ is linear}\}$ & $\{x \oplus y,1\}$ &\\
    $\CloneL_0$   & $\CloneL \cap \CloneR_0$ & $\{x \oplus y\}$ &\\
    $\CloneL_1$   & $\CloneL \cap \CloneR_1$ & $\{x \leftrightarrow y\}$ &\\
    $\CloneL_2$   & $\CloneL \cap \CloneR_2$ & $\{x \oplus y \oplus z\}$ &\\
    $\CloneL_3$   & $\CloneL \cap \CloneD$ & $\{x \oplus y \oplus z \oplus 1 \}$ &\\
    
    $\CloneE$     & $\{f \in \CloneBF \mid f \mbox{ is constant or a conjunction}\}$ & $\{x \land y,0, 1\}$ &\\
    $\CloneE_0$   & $\CloneE \cap \CloneR_0$ & $\{x \land y,0\}$ &\\
    $\CloneE_1$   & $\CloneE \cap \CloneR_1$ & $\{x \land y,1\}$ &\\
    $\CloneE_2$   & $\CloneE \cap \CloneR_2$ & $\{x \land y\}$ &\\
    
    $\CloneV$     & $\{f \in \CloneBF \mid f \mbox{ is constant or a disjunction}\}$ & $\{x \lor y,0,1\}$ &\\
    $\CloneV_0$   & $\CloneV \cap \CloneR_0$ & $\{x \lor y,0\}$ &\\
    $\CloneV_1$   & $\CloneV \cap \CloneR_1$ & $\{x \lor y,1\}$ &\\
    $\CloneV_2$   & $\CloneV \cap \CloneR_2$ & $\{x \lor y\}$ &\\
    
    $\CloneN$     & $\{f \in \CloneBF \mid f \mbox{ is essentially unary}\}$ & $\{\neg x,0,1\}$ &\\ 
    $\CloneN_2$   & $\CloneN \cap \CloneD$ & $\{\neg x\}$ &\\

    $\CloneI$     & $\{f \in \CloneBF \mid f \mbox{ is constant or a projection}\}$ & $\{x,0,1\}$ &\\ 
    $\CloneI_0$   & $\CloneI \cap \CloneR_0$ & $\{x,0\}$ &\\
    $\CloneI_1$   & $\CloneI \cap \CloneR_1$ & $\{x,1\}$ &\\
    $\CloneI_2$   & $\CloneI \cap \CloneR_2$ & $\{x\}$ &\\
    \bottomrule
    \end{tabular}
    \bigskip

\protect\caption{\label{tab:List-of-all}List of all closed classes of Boolean functions
with definitions and bases.\protect \\
($T_{k}^{n}$ denotes the threshold function, $T_{k}^{n}(x_{1},\ldots,x_{n})=1\protect\Longleftrightarrow\sum_{i=1}^{n}x_{i}\geq k$,
and dual($f$)$(x_{1},\ldots,x_{n})=\overline{f(\overline{x_{1}},\ldots,\overline{x_{n}})}$) }
\end{table}

The complexity of numerous problems for $B$-circuits and $B$-formulas
has been classified by the types of functions allowed in $B$ with
help of Post's lattice (see e.g. \citep{reith2000,schnoor2007}),
starting with satisfiability: Analogously to Schaefer's dichotomy
for CNF($\mathcal{S}$)-formulas\emph{ }from 1978\emph{,} Harry R.
Lewis shortly thereafter found a dichotomy for $B$-formulas \citep{lewis1979}:
If $[B]$ contains the function $x\wedge\overline{y}$, \noun{Sat}
is NP-complete, else it is in P.

While for $B$-circuits the complexity of every decision problem solely
depends on $[B]$ (up to AC$^{0}$ isomorphisms), for $B$-formulas
this need not be the case (though it usually is, as for satisfiability
and our connectivity problems, as we will see): The transformation
of a $B$-formula into a $B'$-formula might require an exponential
increase in the formula size even if $[B]=[B']$, as the $B'$-representation
of some function from $B$ may need to use some input variable more
than once \citep{michael2012applicability}. For example, let $h(x,y)=x\wedge\overline{y}$;
then $(\text{x\ensuremath{\wedge}y)}\in[\{h\}]$ since $x\wedge y=h(x,h(x,y))$,
but it is easy to see that there is no shorter $\{h\}$-representation
of $x\wedge y$.\newpage{}

\section{Computational and Structural Dichotomies for Connectivity}

Now we consider the connectivity problems for $B$-formulas and $B$-circuits:
\begin{itemize}
\item \noun{BF-Conn($B$)}: Given a $B$-formula $\phi$, is $G(\phi)$
connected?
\item \noun{st-BF-Conn($B$): }Given a $B$-formula $\phi$ and two solutions
$\boldsymbol{s}$ and $\boldsymbol{t}$, is there a path from $\boldsymbol{s}$
to $\boldsymbol{t}$ in $G(\phi)$?
\end{itemize}
The corresponding problems for circuits are denoted \noun{Circ-Conn($B$)
}resp. \noun{st-Circ-Conn($B$)}.
\begin{thm}
\label{thm:func}Let $B$ be a finite set of Boolean functions.\end{thm}
\begin{enumerate}
\item If $B\subseteq\mathsf{M}$, $B\subseteq\mathsf{L}$, or $B\subseteq\mathsf{\mathsf{S}_{0}}$,
then

\begin{enumerate}
\item \noun{st-Circ-Conn(}$B$\emph{)} and \noun{Circ-Conn(}$B$\emph{)}
are in \noun{P},

\begin{enumerate}
\item \noun{st-BF-Conn(}$B$\emph{)} and \noun{BF-Conn(}$B$\emph{)} are
in \noun{P},
\item the diameter of every function $f\in[B]$ is linear in the number
of variables of $f$.
\end{enumerate}
\item Otherwise,

\begin{enumerate}
\item \noun{st-Circ-Conn(}$B$\emph{)} and \noun{Circ-Conn(}$B$\emph{)}
are $\mathrm{PSPACE}$-complete,
\item \noun{st-BF-Conn(}$B$\emph{)} and \noun{BF-Conn(}$B$\emph{)} are
$\mathrm{PSPACE}$-complete,
\item there are functions $f\in[B]$ such that their diameter is exponential
in the number of variables of $f$.
\end{enumerate}
\end{enumerate}
\end{enumerate}
The proof follows from the Lemmas in the next subsections. By the
following proposition, we can relate the complexity of $B$-formulas
and $B$-circuits.
\begin{prop}
\label{pro:trf}Every $B$-formula $\phi$ can be transformed into
an equivalent $B$-circuit $\mathcal{C}$ in polynomial time.\end{prop}
\begin{proof}
Any $B$-formula is equivalent to a special $B$-circuit where all
function-gates have outdegree at most one: For every variable $x$
of $\phi$ and for every occurrence of a function $f$ in $\phi$
there is a gate in $\mathcal{C}$, labeled with $x$ resp. $f$. It
is clear how to connect the gates.
\end{proof}

\subsection{\label{sub:The-easy-cases}The Easy Side of the Dichotomy}
\begin{lem}
\label{lem:M}If $B\subseteq\mathsf{M}$, the solution graph of any
$n$-ary function $f\in[B]$ is connected, and $d_{f}(\boldsymbol{a},\boldsymbol{b})=|\boldsymbol{a}-\boldsymbol{b}|\leq n$
for any two solutions $\boldsymbol{a}$ and $\boldsymbol{b}$.\end{lem}
\begin{proof}
Table \ref{tab:List-of-all} shows that $f$ is monotone in this case.
Thus, either $f=0$, or $(1,\ldots,1)$ must be a solution, and every
other solution $\boldsymbol{a}$ is connected to $(1,\ldots,1)$ in
$G(\phi)$ since $(1,\ldots,1)$ can be reached by flipping the variables
assigned 0 in $\boldsymbol{a}$ one at a time to 1. Further, if $\boldsymbol{a}$
and $\boldsymbol{b}$ are solutions, $\boldsymbol{b}$ can be reached
from $\boldsymbol{a}$ in $|\boldsymbol{a}-\boldsymbol{b}|$ steps
by first flipping all variables that are assigned 0 in $\boldsymbol{a}$
and 1 in $\boldsymbol{b}$, and then flipping all variables that are
assigned 1 in $\boldsymbol{a}$ and 0 in $\boldsymbol{b}$.\end{proof}
\begin{lem}
If $B\subseteq\mathsf{S}_{0}$, the solution graph of any function
$f\in[B]$ is connected, and $d_{f}(\boldsymbol{a},\boldsymbol{b})\leq|\boldsymbol{a}-\boldsymbol{b}|+2$
for any two solutions $\boldsymbol{a}$ and $\boldsymbol{b}$.\end{lem}
\begin{proof}
Since $f$ is 0-separating, there is an $i$ such that $a_{i}=0$
for every vector $\boldsymbol{a}$ with $f(\boldsymbol{a})=0$, thus
every $\boldsymbol{b}$ with $b_{i}=1$ is a solution. It follows
that every solution $\boldsymbol{t}$ can be reached from any solution
$\boldsymbol{s}$ in at most $|\boldsymbol{s}-\boldsymbol{t}|+2$
steps by first flipping the $i$-th variable from 0 to 1 if necessary,
then flipping all other variables in which $\boldsymbol{s}$ and $\boldsymbol{t}$
differ, and finally flipping back the $i$-th variable if necessary.\end{proof}
\begin{lem}
\label{lem:L}If $B\subseteq\mathsf{L}$,\end{lem}
\begin{enumerate}
\item \noun{st-Circ-Conn(}$B$\emph{)} and \noun{Circ-Conn(}$B$\emph{)}
are in P,

\begin{enumerate}
\item \noun{st-BF-Conn(}$B$\emph{)} and \noun{BF-Conn(}$B$\emph{)} are
in P,
\item for any function $f\in[B]$, $d_{f}(\boldsymbol{a},\boldsymbol{b})=|\boldsymbol{a}-\boldsymbol{b}|$
for any two solutions $\boldsymbol{a}$ and $\boldsymbol{b}$ that
lie in the same connected component of $G(\phi)$.
\end{enumerate}
\end{enumerate}
\begin{proof}
Since every function $f\in\mathsf{L}$ is linear, $f(x_{1},\ldots,x_{n})=x_{i_{1}}\oplus\ldots\oplus x_{i_{m}}\oplus c$,
and any two solutions $\boldsymbol{s}$ and $\boldsymbol{t}$ are
connected iff they differ only in fictive variables: If $\boldsymbol{s}$
and $\boldsymbol{t}$ differ in at least one non-fictive variable
(i.e., an $x_{i}\in\{x_{i_{1}},\ldots,x_{i_{m}}\}$), to reach $\mathbf{t}$
from $\mathbf{s}$, $x_{i}$ must be flipped eventually, but for every
solution $\boldsymbol{a}$, any vector $\boldsymbol{b}$ that differs
from $\boldsymbol{a}$ in exactly one non-fictive variable is no solution.
If $\boldsymbol{s}$ and $\boldsymbol{t}$ differ only in fictive
variables, $\boldsymbol{t}$ can be reached from $\boldsymbol{s}$
in $|\boldsymbol{s}-\boldsymbol{t}|$ steps by flipping one by one
the variables in which they differ.

Since $\{x\oplus y,1\}$ is a base of $\mathsf{L}$, every $B$-circuit
$\mathcal{C}$ can be transformed in polynomial time into an equivalent
$\{x\oplus y,1\}$-circuit $\mathcal{C}'$ by replacing each gate
of $\mathcal{C}$ with an equivalent $\{x\oplus y,1\}$-circuit. Now
one can decide in polynomial time whether a variable $x_{i}$ is fictive
by checking for $\mathcal{C}'$ whether the number of ``backward
paths'' from the output gate to gates labeled with $x_{i}$ is odd,
so \noun{st-Circ-Conn(}$B$) is in P.

$G(\mathcal{C})$ is connected iff at most one variable is non-fictive,
thus \noun{Circ-Conn(}$B$) is in P.

By Proposition \ref{pro:trf}, \noun{st-BF-Conn(}$B$) and \noun{BF-Conn(}$B$)
are in P also.
\end{proof}
This completes the proof of the easy side of the dichotomy.

\subsection{\label{sub:The-hard-cases}The Hard Side of the Dichotomy}
\begin{prop}
\noun{st-Circ-Conn(}$B$\emph{)} and \noun{Circ-Conn(}$B$\emph{)},
as well as \noun{st-BF-Conn(}$B$\emph{)} and \noun{BF-Conn(}$B$\emph{)},
are in $\mathrm{PSPACE}$ for any finite set \noun{$B$} of Boolean
functions\noun{.}\end{prop}
\begin{proof}
This follows as in Lemma 3.6 of \citep{gop} (see \prettyref{lem:-gopa}).
\end{proof}
An inspection of Post's lattice shows that if $B\nsubseteq\mathsf{M}$,
$B\nsubseteq\mathsf{L}$, and $B\nsubseteq\mathsf{S}_{0}$, then $[B]\supseteq\mathsf{S}_{12}$,
$[B]\supseteq\mathsf{D}_{1}$, or $[B]\supseteq\mathsf{S}_{02}^{k}\,\forall k\geq2$,
so we have to prove $\mathrm{PSPACE}$-completeness and show the existence
of $B$-formulas with an exponential diameter in these cases.

In the proofs, we will use the following notation: We write $\boldsymbol{x}=\boldsymbol{c}$
or $\boldsymbol{x}=c_{1}\cdots c_{n}$ for $(x_{1}=c_{1})\wedge\cdots\wedge(x_{n}=c_{n})$,
where $\boldsymbol{c}=(c_{1},\ldots,c_{n})$ is a vector of constants;
e.g., $\boldsymbol{x}=\boldsymbol{0}$ means $\overline{x}_{1}\wedge\cdots\wedge\overline{x}_{n}$,
and $\boldsymbol{x}=101$ means $x_{1}\wedge\overline{x}_{2}\wedge x_{3}$.
Further, we use $\boldsymbol{x}\in\left\{ \boldsymbol{a},\boldsymbol{b},\ldots\right\} $
for $(\boldsymbol{x}=\boldsymbol{a})\vee(\boldsymbol{x}=\boldsymbol{b})\vee\ldots$.
Also, we write $\psi(\boldsymbol{\overline{x}})$ for $\psi(\overline{x}_{1},\ldots,\overline{x}_{n})$.
If we have two vectors of Boolean values $\boldsymbol{a}$ and $\boldsymbol{b}$
of length $n$ and $m$ resp., we write $\boldsymbol{a}\cdot\boldsymbol{b}$
for their concatenation $(a_{1},\ldots,a_{n},b_{1},\ldots b_{m})$.

All hardness proofs are by reductions from the problems for 1-reproducing
3-CNF-formulas, which are $\mathrm{PSPACE}$-complete by the following
proposition.
\begin{prop}
For 1-reproducing 3-CNF-formulas, the problems \noun{st-Conn }and
\noun{Conn} are $\mathrm{PSPACE}$-complete.\end{prop}
\begin{proof}
In the $\mathrm{PSPACE}$-hardness proof for CNF($S_{3}$)-formulas
(Lemma 3.6 of \citep{gop}, see \prettyref{lem:-gopa}), two satisfying
assignments $\boldsymbol{s}$ and $\boldsymbol{t}$ to the constructed
formula $\phi$ are known, so we can construct a connectivity-equivalent
1-reproducing 3-CNF-formula $\psi$, e.g. as $\psi(\boldsymbol{x})=\phi(x_{1}\oplus s_{1}\oplus1,\ldots,x_{n}\oplus s_{n}\oplus1)$,
and then check connectivity for $\psi$ instead of $\phi$.\qed\end{proof}
\begin{lem}
\label{lem:s12}If $[B]\supseteq\mathsf{S}_{12}$,\end{lem}
\begin{enumerate}
\item \noun{st-BF-Conn(}$B$\emph{)} and \noun{BF-Conn(}$B$\emph{)} are
$\mathrm{PSPACE}$-complete,

\begin{enumerate}
\item \noun{st-Circ-Conn(}$B$\emph{)} and \noun{Circ-Conn(}$B$\emph{)}
are $\mathrm{PSPACE}$-complete,
\item for $n\geq3$, there is an $n$-ary function $f\in[B]$ with diameter
of at least $2^{\left\lfloor \frac{n-1}{2}\right\rfloor }$.
\end{enumerate}
\end{enumerate}
\begin{proof}
1. We reduce the problems for 1-reproducing 3-CNF-formulas to the
ones for $B$-formulas: We map a 1-reproducing 3-CNF-formula $\phi$
and two solutions $\boldsymbol{s}$ and $\boldsymbol{t}$ of $\phi$
to a $B$-formula $\phi'$ and two solutions $\boldsymbol{s'}$ and
$\boldsymbol{t'}$ of $\phi'$ such that $\boldsymbol{s'}$ and $\boldsymbol{t'}$
are connected in $G(\phi')$ iff $\boldsymbol{s}$ and $\boldsymbol{t}$
are connected in $G(\phi)$, and such that $G(\phi')$ is connected
iff $G(\phi)$ is connected.

While the construction of $\phi'$ is quite easy for this lemma, the
construction for the next two lemmas is analogous but more intricate,
so we proceed carefully in two steps, which we will adapt in the next
two proofs: In the first step, we give a transformation $T$ that
transforms any 1-reproducing formula $\psi$ into a connectivity-equivalent
formula $T_{\psi}\in\mathsf{S}_{12}$ built from the standard connectives.
Since $\mathsf{S}_{12}\subseteq[B]$, we can express $T_{\psi}$ as
a $B$-formula $T_{\psi}^{*}$. Now if we would apply $T$ to $\phi$
directly, we would know that $T_{\phi}$ can be expressed as a $B$-formula.
However, this could lead to an exponential increase in the formula
size (see \prettyref{sec:Cir}), so we have to show how to construct
the $B$-formula in polynomial time. For this, in the second step,
we construct a $B$-formula $\phi'$ directly from $\phi$ (by applying
$T$ to the clauses and the $\wedge$'s individually), and then show
that $\phi'$ is equivalent to $T_{\phi}$; thus we know that $\phi'$
is connectivity-equivalent to $\phi$.

\emph{Step 1.} From Table \ref{tab:List-of-all}, we find that $\mathsf{S}_{12}=\mathsf{S}_{1}\cap\mathsf{R}_{2}=\mathsf{S}_{1}\cap\mathsf{R}_{0}\cap\mathsf{R}_{1}$,
so we have to make sure that $T_{\psi}$ is 1-seperating, 0-reproducing,
and 1-reproducing. Let 
\[
T_{\psi}=\psi\wedge y,
\]
where $y$ is a new variable.

All solutions $\boldsymbol{a}$ of $T_{\psi}(\boldsymbol{x},y)$ have
$a_{n+1}=1$, so $T_{\psi}$ is 1-seperating and 0-reproducing; also,
$T_{\psi}$ is still 1-reproducing. Further, for any two solutions
$\boldsymbol{s}$ and $\boldsymbol{t}$ of $\psi(\boldsymbol{x})$,
$\boldsymbol{s}'=\boldsymbol{s}\cdot1$ and $\boldsymbol{t}'=\boldsymbol{t}\cdot1$
are solutions of $T_{\psi}(\boldsymbol{x},y)$, and it is easy to
see that they are connected in $G(T_{\psi})$ iff $\boldsymbol{s}$
and $\boldsymbol{t}$ are connected in $G(\psi)$, and that $G(T_{\psi})$
is connected iff $G(\psi)$ is connected.

\emph{Step 2.} The idea is to parenthesize the conjunctions of $\phi$
such that we get a tree of $\wedge$'s of depth logarithmic in the
size of $\phi$, and then to replace each clause and each $\wedge$
with an equivalent $B$-formula. This can increase the formula size
by only a polynomial in the original size even if the $B$-formula
equivalent to $\wedge$ uses some input variable more than once.

Let $\phi=C_{1}\wedge\cdots\wedge C_{n}$ be a 1-reproducing 3-CNF-formula.
Since $\phi$ is 1-reproducing, every clause $C_{i}$ of $\phi$ is
itself 1-reproducing, and we can express $T_{C_{i}}$ through a $B$-formula
$T_{C_{i}}^{*}$. Also, we can express $T_{u\wedge v}$ through a
$B$-formula $T_{u\wedge v}^{*}$ since $\wedge$ is 1-reproducing;
we write $T_{\wedge}(\psi_{1},\psi_{2})$ for the formula obtained
from $T_{u\wedge v}$ by substituting the formula $\psi_{1}$ for
$u$ and $\psi_{2}$ for $v$, and similarly write $T_{\wedge}^{*}(\psi_{1},\psi_{2})$
for the formula obtained from $T_{u\wedge v}^{*}$ in this way. We
let $\phi'=$\noun{Tr$(\phi)$}, where \noun{Tr} is the following
recursive algorithm that takes a CNF-formula as input:\\
\\
Algorithm \noun{Tr}$\left(\psi_{1}\wedge\cdots\wedge\psi_{m}\right)$\end{proof}
\begin{itemize}[label= ]
\item If $m=1$, return $T_{\psi_{1}}^{*}$.
\item Else return $T_{\wedge}^{*}\left(\textsc{Tr}(\psi_{1}\wedge\cdots\wedge\psi_{\left\lfloor m/2\right\rfloor }),\textsc{Tr}(\psi_{\left\lfloor m/2\right\rfloor +1}\wedge\cdots\wedge\psi_{m})\right)$.\end{itemize}
\begin{proof}
Since the recursion terminates after a number of steps logarithmic
in the number of clauses of $\phi$, and every step increases the
total formula size by only a constant factor, the algorithm runs in
polynomial time. We show $\phi'\equiv T_{\phi}$ by induction on $m$.
For $m=1$ this is clear. For the induction step, we have to show
$T_{\wedge}^{*}(T_{\psi_{1}},T_{\psi_{2}})\equiv T_{\psi_{1}\wedge\psi_{2}}$,
but since $T_{\wedge}(\psi_{1},\psi_{2})\equiv T_{\wedge}^{*}(\psi_{1},\psi_{2})$,
it suffices to show that $T_{\wedge}(T_{\psi_{1}},T_{\psi_{2}})\equiv T_{\psi_{1}\wedge\psi_{2}}$:

\[
T_{\wedge}(T_{\psi_{1}},T_{\psi_{2}})=(\psi_{1}\wedge y)\wedge(\psi_{2}\wedge y)\wedge y\equiv\psi_{1}\wedge\psi_{2}\wedge y=T_{\psi_{1}\wedge\psi_{2}}.
\]

2. This follows from 1.\ by Proposition \ref{pro:trf}.

3. By \prettyref{lem:gop}, there is an 1-reproducing $(n-1)$-ary
function $f$ with diameter of at least $2^{\left\lfloor \frac{n-1}{2}\right\rfloor }$.
Let $f$ be represented by a formula $\phi$; then, $T_{\phi}$ represents
an $n$-ary function of the same diameter in $\mathsf{S}_{12}$.\end{proof}
\begin{lem}
If $[B]\supseteq\mathsf{D}_{1}$,\end{lem}
\begin{enumerate}
\item \noun{st-BF-Conn(}$B$\emph{)} and \noun{BF-Conn(}$B$\emph{)} are
$\mathrm{PSPACE}$-complete,

\begin{enumerate}
\item \noun{st-Circ-Conn(}$B$\emph{)} and \noun{Circ-Conn(}$B$\emph{)}
are $\mathrm{PSPACE}$-complete,
\item for $n\geq5$, there is an $n$-ary function $f\in[B]$ with diameter
of at least $2^{\left\lfloor \frac{n-3}{2}\right\rfloor }$.
\end{enumerate}
\end{enumerate}
\begin{proof}
1.$\;$As noted, we adapt the two steps from the previous proof.

\emph{Step 1. }Since $\mathsf{D}_{1}=\mathsf{D}\cap\mathsf{R}_{0}\cap\mathsf{R}_{1}$,
$T_{\psi}$ must be self-dual, 0-reproducing, and 1-reproducing. For
clarity, we first construct an intermediate formula $T_{\psi}^{\sim}\in\mathsf{D}_{1}$
whose solution graph has an additional component, then we eliminate
that component.

For $\psi(\boldsymbol{x})$, let
\[
T_{\psi}^{\sim}=\left(\psi(\boldsymbol{x})\wedge(\boldsymbol{y}=\boldsymbol{1})\right)\vee\left(\overline{\psi(\boldsymbol{\overline{x}})}\wedge(\boldsymbol{y}=\boldsymbol{0})\right)\vee\left(\boldsymbol{y}\in\left\{ 100,010,001\right\} \right),
\]
where $\boldsymbol{y}=(y_{1},y_{2},y_{3})$ are three new variables.

$T_{\psi}^{\sim}$ is self-dual: for any solution ending with 111
(satisfying the first disjunct), the inverse vector is no solution;
similarly, for any solution ending with 000 (satisfying the second
disjunct), the inverse vector is no solution; finally, all vectors
ending with 100, 010, or 001 are solutions and their inverses are
no solutions. Also, $T_{\psi}^{\sim}$ is still 1-reproducing, and
it is 0-reproducing (for the second disjunct note that $\overline{\psi(\overline{0\cdots0})}\equiv\overline{\psi(1\cdots1)}\equiv0$).

Further, every solution $\boldsymbol{a}$ of $\psi$ corresponds to
a solution $\boldsymbol{a}\cdot111$ of $T_{\psi}^{\sim}$, and for
any two solutions $\boldsymbol{s}$ and $\boldsymbol{t}$ of $\psi$,
$\boldsymbol{s}'=\boldsymbol{s}\cdot111$ and $\boldsymbol{t}'=\boldsymbol{t}\cdot111$
are connected in $G(T_{\psi}^{\sim})$ iff $\boldsymbol{s}$ and $\boldsymbol{t}$
are connected in $G(\psi)$: The ``if'' is clear, for the ``only
if'' note that since there are no solutions of $T_{\psi}^{\sim}$
ending with 110, 101, or 011, every solution of $T_{\psi}^{\sim}$
not ending with 111 differs in at least two variables from the solutions
that do.

Observe that exactly one connected component is added in $G(T_{\psi}^{\sim})$
to the components corresponding to those of $G(\psi)$: It consists
of all solutions ending with 000, 100, 010, or 001 (any two vectors
ending with 000 are connected e.g. via those ending with 100). It
follows that $G(T_{\psi}^{\sim})$ is always unconnected. To fix this,
we modify $T_{\psi}^{\sim}$ to $T_{\psi}$ by adding $1\cdots1\cdot110$
as a solution, thereby connecting $1\cdots1\cdot111$ (which is always
a solution since $T_{\psi}^{\sim}$ is 1-reproducing) with $1\cdots1\cdot100$,
and thereby with the additional component of $T_{\psi}$. To keep
the function self-dual, we must in turn remove $0\cdots0\cdot001$,
which does not alter the connectivity. Formally,

\begin{eqnarray}
T_{\psi} & = & \left(T_{\psi}^{\sim}\vee\left((\boldsymbol{x}=\boldsymbol{1})\wedge(\boldsymbol{y}=110)\right)\right)\wedge\neg\left((\boldsymbol{x}=\boldsymbol{0})\wedge(\boldsymbol{y}=001)\right)\label{eq:q1}\\
 & = & \left(\psi(\boldsymbol{x})\wedge(\boldsymbol{y}=\boldsymbol{1})\right)\vee\left(\overline{\psi(\boldsymbol{\overline{x}})}\wedge(\boldsymbol{y}=\boldsymbol{0})\right)\nonumber \\
 &  & \vee\left(\boldsymbol{y}\in\left\{ 100,010,001\right\} \wedge\neg((\boldsymbol{x}=\boldsymbol{0})\wedge(\boldsymbol{y}=001))\right)\nonumber \\
 &  & \vee((\boldsymbol{x}=\boldsymbol{1})\wedge(\boldsymbol{y}=110)).\nonumber 
\end{eqnarray}

\begin{figure}[!h]
\begin{centering}
\includegraphics[scale=0.36]{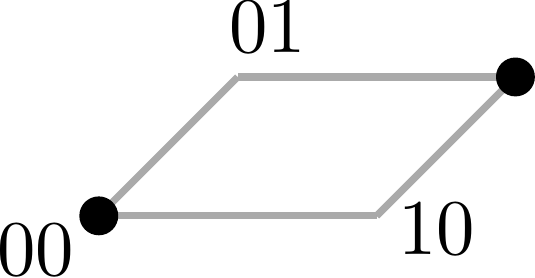}$\qquad$\includegraphics[scale=0.72]{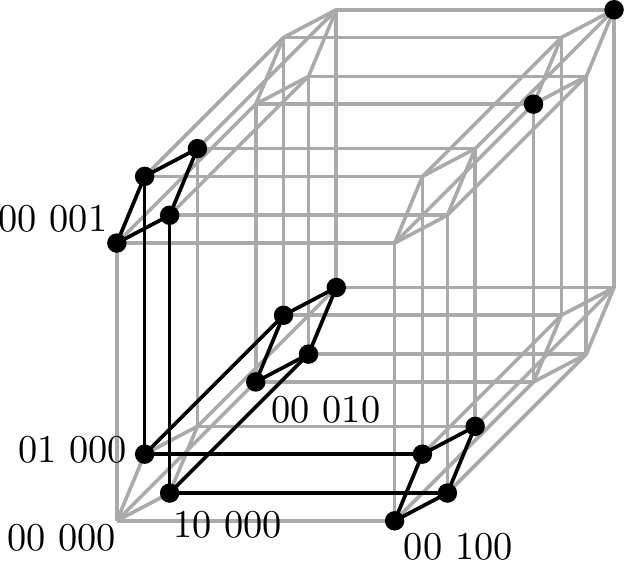}$\qquad$\includegraphics[scale=0.72]{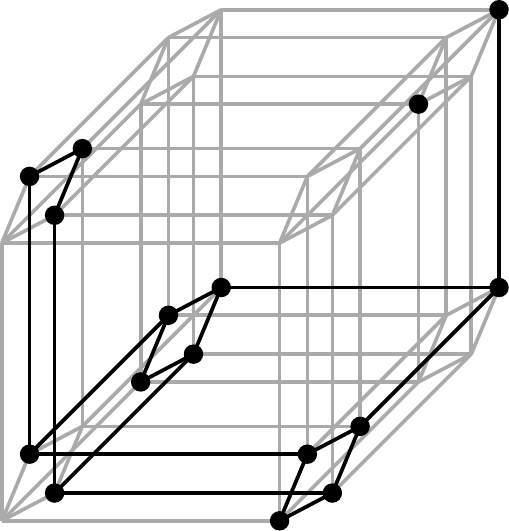}
\par\end{centering}

\protect\caption{An example for the transformation. Left: $\psi=\left(x_{1}\vee\overline{x_{2}}\right)\wedge\left(\overline{x_{1}}\vee x_{2}\right)$,
center: $T_{\psi}^{\sim}$, right: $T_{\psi}$. The \textquotedblleft axis
vertices\textquotedblright{} are labeled in the first two graphs.}
\end{figure}

Now $G(T_{\psi})$ is connected iff $G(\psi)$ is connected.

\emph{Step 2.} Again, we use the algorithm \noun{Tr} from the previous
proof to transform any 1-reproducing 3-CNF-formula $\phi$ into a
$B$-formula $\phi'$ equivalent to $T_{\phi}$, but with the definition
\prettyref{eq:q1} of $T$. Again, we have to show $T_{\wedge}(T_{\psi_{1}},T_{\psi_{2}})\equiv T_{\psi_{1}\wedge\psi_{2}}$.
Here, 
\begin{eqnarray*}
T_{\wedge}(T_{\psi_{1}},T_{\psi_{2}}) & = & \left(T_{\psi_{1}}\wedge T_{\psi_{2}}\wedge(\boldsymbol{y}=\boldsymbol{1})\right)\vee\left(\overline{\overline{T_{\psi_{1}}}\wedge\overline{T_{\psi_{2}}}}\wedge(\boldsymbol{y}=\boldsymbol{0})\right)\\
 &  & \vee\left(\boldsymbol{y}\in\left\{ 100,010,001\right\} \wedge\neg\left(\overline{T_{\psi_{1}}}\wedge\overline{T_{\psi_{2}}}\wedge(\boldsymbol{y}=001)\right)\right)\\
 &  & \vee\left(T_{\psi_{1}}\wedge T_{\psi_{2}}\wedge(\boldsymbol{y}=110)\right).
\end{eqnarray*}
We consider the parts of the formula in turn: For any formula $\xi$
we have $T_{\xi}(\boldsymbol{x}_{\xi})\wedge(\boldsymbol{y}=\boldsymbol{1})\equiv\xi(\boldsymbol{x}_{\xi})\wedge(\boldsymbol{y}=\boldsymbol{1})$
and $T_{\xi}(\boldsymbol{x}_{\xi})\wedge(\boldsymbol{y}=\boldsymbol{0})\equiv\overline{\psi(\overline{\boldsymbol{x}_{\xi}})}\wedge(\boldsymbol{y}=\boldsymbol{0})$,
where $\boldsymbol{x}_{\xi}$ denotes the variables of $\xi$. Using
$\overline{\overline{T_{\psi_{1}}(\boldsymbol{x}_{\psi_{1}})}\wedge\overline{T_{\psi_{2}}(\boldsymbol{x}_{\psi_{2}})}}\wedge(\boldsymbol{y}=\boldsymbol{0})=\left(T_{\psi_{1}}(\boldsymbol{x}_{\psi_{1}})\vee T_{\psi_{2}}(\boldsymbol{x}_{\psi_{2}})\right)\wedge(\boldsymbol{y}=\boldsymbol{0})$,
the first line becomes 
\[
\left(\psi_{1}(\boldsymbol{x}_{\psi_{1}})\wedge\psi_{2}(\boldsymbol{x}_{\psi_{2}})\wedge(\boldsymbol{y}=\boldsymbol{1})\right)\vee\left(\left(\overline{\psi_{1}(\overline{\boldsymbol{x}_{\psi_{1}}})\wedge\psi_{2}(\overline{\boldsymbol{x}_{\psi_{2}}})}\right)\wedge(\boldsymbol{y}=\boldsymbol{0})\right).
\]
For the second line, we observe 
\begin{eqnarray*}
\overline{T_{\psi}(\boldsymbol{x}_{\psi})} & \equiv & \left(\overline{\psi(\boldsymbol{x}_{\psi})}\vee\neg(\boldsymbol{y}=\boldsymbol{1})\right)\wedge\left(\psi(\boldsymbol{\overline{x}_{\psi}})\vee\neg(\boldsymbol{y}=\boldsymbol{0})\right)\\
 &  & \wedge\left(\boldsymbol{y}\notin\left\{ 100,010,001\right\} \vee\left((\boldsymbol{x}_{\psi}=\boldsymbol{0})\wedge(\boldsymbol{y}=001)\right)\right)\\
 &  & \wedge(\neg(\boldsymbol{x}_{\psi}=\boldsymbol{1})\vee\overline{(\boldsymbol{y}=110)}),
\end{eqnarray*}
thus $\overline{T_{\psi}(\boldsymbol{x}_{\psi})}\wedge(\boldsymbol{y}=001)\equiv(\boldsymbol{x}_{\psi}=\boldsymbol{0})\wedge(\boldsymbol{y}=001)$,
and the second line becomes
\[
\vee\left(\boldsymbol{y}\in\left\{ 100,010,001\right\} \wedge\neg\left((\boldsymbol{x}_{\psi_{1}}=\boldsymbol{0})\wedge(\boldsymbol{x}_{\psi_{2}}=\boldsymbol{0})\wedge(\boldsymbol{y}=001)\right)\right).
\]

Since $T_{\psi}(\boldsymbol{x}_{\psi})\wedge(\boldsymbol{y}=110)\equiv(\boldsymbol{x}_{\psi}=\boldsymbol{1})\wedge(\boldsymbol{y}=110)$
for any $\psi$, the third line becomes
\[
\vee\left((\boldsymbol{x}_{\psi_{1}}=\boldsymbol{1})\wedge(\boldsymbol{x}_{\psi_{2}}=\boldsymbol{1})\wedge(\boldsymbol{y}=110)\right).
\]
Now $T_{\wedge}(T_{\psi_{1}},T_{\psi_{2}})$ equals
\begin{eqnarray*}
T_{\psi_{1}\wedge\psi_{2}} & = & \left(\psi_{1}(\boldsymbol{x}_{\psi_{1}})\wedge\psi_{2}(\boldsymbol{x}_{\psi_{2}})\wedge(\boldsymbol{y}=\boldsymbol{1})\right)\vee\left(\overline{\psi_{1}(\overline{\boldsymbol{x}_{\psi_{1}}})\wedge\psi_{2}(\overline{\boldsymbol{x}_{\psi_{2}}})}\wedge(\boldsymbol{y}=\boldsymbol{0})\right)\\
 &  & \vee\left(\boldsymbol{y}\in\left\{ 100,010,001\right\} \wedge\neg\left((\boldsymbol{x}_{\psi_{1}}=\boldsymbol{0})\wedge(\boldsymbol{x}_{\psi_{2}}=\boldsymbol{0})\wedge(\boldsymbol{y}=001)\right)\right)\\
 &  & \vee\left((\boldsymbol{x}_{\psi_{1}}=\boldsymbol{1})\wedge(\boldsymbol{x}_{\psi_{2}}=\boldsymbol{1})\wedge(\boldsymbol{y}=110)\right).
\end{eqnarray*}

2. This follows from 1.\ by Proposition \ref{pro:trf}.

3. By \prettyref{lem:gop} there is an 1-reproducing $(n-3)$-ary
function $f$ with diameter of at least $2^{\left\lfloor \frac{n-3}{2}\right\rfloor }$.
Let $f$ be represented by a formula $\phi$; then, $T_{\phi}$ represents
an $n$-ary function of the same diameter in $\mathsf{D}_{1}$.\end{proof}
\begin{lem}
If $[B]\supseteq\mathsf{S}_{02}^{k}$ for any $k\geq2$,\end{lem}
\begin{enumerate}
\item \noun{st-BF-Conn(}$B$\emph{)} and \noun{BF-Conn(}$B$\emph{)} are
$\mathrm{PSPACE}$-complete,

\begin{enumerate}
\item \noun{st-Circ-Conn(}$B$\emph{)} and \noun{Circ-Conn(}$B$\emph{)}
are $\mathrm{PSPACE}$-complete,
\item for $n\geq k+4$, there is an $n$-ary function $f\in[B]$ with diameter
of at least $2^{\left\lfloor \frac{n-k-2}{2}\right\rfloor }$.
\end{enumerate}
\end{enumerate}
\begin{proof}
1. \emph{Step 1. }Since $\mathsf{S}_{02}^{k}=\mathsf{S}_{0}^{k}\cap\mathsf{R}_{0}\cap\mathsf{R}_{1}$,
$T_{\psi}$ must be 0-separating of degree $k$, 0-reproducing, and
1-reproducing. As in the previous proof, we construct an intermediate
formula $T_{\psi}^{\sim}$. For $\psi(\boldsymbol{x})$, let 
\[
T_{\psi}^{\sim}=\left(\psi\wedge y\wedge(\boldsymbol{z}=\boldsymbol{0})\right)\vee(|\boldsymbol{z}|>1),
\]
where $y$ and $\boldsymbol{z}=(z_{1},\ldots,z_{k+1})$ are new variables.

$T_{\psi}^{\sim}(\boldsymbol{x},y,\boldsymbol{z})$ is 0-separating
of degree $k$, since all vectors that are no solutions of $T_{\psi}^{\sim}$
have $|\boldsymbol{z}|\leq1$, i.e. $\boldsymbol{z}\in\left\{ 0\cdots0,10\cdots0,010\cdots0,\ldots,0\cdots01\right\} \subset\{0,1\}^{k+1}$,
and thus any $k$ of them have at least one common variable assigned
0. Also, $T_{\psi}^{\sim}$ is 0-reproducing and still 1-reproducing.

Further, for any two solutions $\boldsymbol{s}$ and $\boldsymbol{t}$
of $\psi(\boldsymbol{x})$, $\boldsymbol{s}'=\boldsymbol{s}\cdot1\cdot0\cdots0$
and $\boldsymbol{t}'=\boldsymbol{t}\cdot1\cdot0\cdots0$ are solutions
of $T_{\psi}^{\sim}(\boldsymbol{x},y,\boldsymbol{z})$ and are connected
in $G(T_{\psi}^{\sim})$ iff $\boldsymbol{s}$ and $\boldsymbol{t}$
are connected in $G(\psi)$.

But again, we have produced an additional connected component (consisting
of all solutions with $|\boldsymbol{z}|>1$). To connect it to a component
corresponding to one of $\psi$, we add $1\cdots1\cdot1\cdot10\cdots0$
as a solution,
\begin{eqnarray*}
T_{\psi} & = & \left(\psi\wedge y\wedge(\boldsymbol{z}=\boldsymbol{0})\right)\vee(|\boldsymbol{z}|>1)\vee\left((\boldsymbol{x}=\boldsymbol{1})\wedge y\wedge(\boldsymbol{z}=10\cdots0)\right).
\end{eqnarray*}
Now $G(T_{\psi})$ is connected iff $G(\psi)$ is connected.

\emph{Step 2.} Again we show that the algorithm \noun{Tr} works in
this case. Here,
\begin{eqnarray*}
T_{\wedge}(T_{\psi_{1}},T_{\psi_{2}}) & = & \left(T_{\psi_{1}}(\boldsymbol{x}_{\psi_{1}})\wedge T_{\psi_{2}}(\boldsymbol{x}_{\psi_{2}})\wedge y\wedge(\boldsymbol{z}=\boldsymbol{0})\right)\vee(|\boldsymbol{z}|>1)\\
 &  & \vee\left(T_{\psi_{1}}(\boldsymbol{x}_{\psi_{1}})\wedge T_{\psi_{2}}(\boldsymbol{x}_{\psi_{2}})\wedge y\wedge(\boldsymbol{z}=10\cdots0)\right).
\end{eqnarray*}
Since $T_{\psi}(\boldsymbol{x}_{\psi})\wedge y\wedge(\boldsymbol{z}=\boldsymbol{0})\equiv\psi(\boldsymbol{x}_{\psi})\wedge y\wedge(\boldsymbol{z}=\boldsymbol{0})$
and $T_{\psi}(\boldsymbol{x}_{\psi})\wedge y\wedge(\boldsymbol{z}=10\cdots0)\equiv(\boldsymbol{x}_{\psi}=1)\wedge y\wedge(\boldsymbol{z}=10\cdots0)$
for any $\psi$, this is equivalent to
\begin{eqnarray*}
T_{\psi_{1}\wedge\psi_{2}} & = & \left(\psi_{1}(\boldsymbol{x}_{\psi_{1}})\wedge\psi_{2}(\boldsymbol{x}_{\psi_{2}})\wedge y\wedge(\boldsymbol{z}=\boldsymbol{0})\right)\vee(|\boldsymbol{z}|>1)\\
 &  & \vee\left(\boldsymbol{x}_{\psi_{1}}\wedge\boldsymbol{x}_{\psi_{2}}\wedge y\wedge(\boldsymbol{z}=10\cdots0)\right).
\end{eqnarray*}

2. This follows from 1.\ by Proposition \ref{pro:trf}.

3. By \prettyref{lem:gop} there is an 1-reproducing $(n-k-2)$-ary
function $f$ with diameter of at least $2^{\left\lfloor \frac{n-k-2}{2}\right\rfloor }$.
Let $f$ be represented by a formula $\phi$; then, $T_{\phi}$ represents
an $n$-ary function of the same diameter in $\mathsf{S}_{02}^{k}$.
\end{proof}
This completes the proof of \prettyref{thm:func}.

\section{\label{sub:Quantified-case} The Connectivity of Quantified Formulas}
\begin{defn}
A \emph{quantified $B$-formula} $\phi$ (in prenex normal form) is
an expression of the form 
\[
Q_{1}y_{1}\cdots Q_{m}y_{m}\varphi(y_{1},\ldots,y_{m},x_{1},\ldots,x_{n}),
\]
where $\varphi$ is a $B$-formula, and $Q_{1},\ldots,Q_{m}\in\{\exists,\forall\}$
are quantifiers. The solution graph $G(\phi)$ only involves the free
variables$x_{1},\ldots,x_{n}$.
\end{defn}
For quantified $B$-formulas, we define the connectivity problems
\begin{itemize}
\item \noun{QBF-Conn($B$)}: Given a quantified $B$-formula $\phi$, is
$G(\phi)$ connected?
\item \noun{st-QBF-Conn($B$): }Given a quantified $B$-formula $\phi$
and two solutions $\boldsymbol{s}$ and $\boldsymbol{t}$, is there
a path from $\boldsymbol{s}$ to $\boldsymbol{t}$ in $G(\phi)$?\end{itemize}
\begin{thm}
\label{thm:quan}Let $B$ be a finite set of Boolean functions.\end{thm}
\begin{enumerate}
\item If $B\subseteq\mathsf{M}$ or $B\subseteq\mathsf{L}$, then

\begin{enumerate}
\item \noun{st-QBF-Conn(}$B$\emph{)} and \noun{QBF-Conn(}$B$\emph{)} are
in \noun{P},

\begin{enumerate}
\item the diameter of every quantified $B$-formula is linear in the number
of free variables.
\end{enumerate}
\item Otherwise,

\begin{enumerate}
\item \noun{st-QBF-Conn(}$B$\emph{)} and \noun{QBF-Conn(}$B$\emph{)} are
$\mathrm{PSPACE}$-complete,
\item there are quantified $B$-formulas with at most one quantifier such
that their diameter is exponential in the number of free variables.
\end{enumerate}
\end{enumerate}
\end{enumerate}
\begin{proof}
1. For $B\subseteq\mathsf{M}$, any quantified $B$-formula $\phi$
represents a monotone function: Using $\exists y\psi(y,\boldsymbol{x})=\psi(0,\boldsymbol{x})\vee\psi(1,\boldsymbol{x})$
and $\forall y\psi(y,\boldsymbol{x})=\psi(0,\boldsymbol{x})\wedge\psi(1,\boldsymbol{x})$
recursively, we can transform $\phi$ into an equivalent $\mathsf{M}$-formula
since $\wedge$ and $\vee$ are monotone. Thus as in \prettyref{lem:M},
\noun{st-QBF-Conn(}$B$) and \noun{QBF-Conn(}$B$) are trivial, and
$d_{f}(\boldsymbol{a},\boldsymbol{b})=|\boldsymbol{a}-\boldsymbol{b}|$
for any two solutions $\boldsymbol{a}$ and $\boldsymbol{b}$.

For a quantified $B$-formula $\phi=Q_{1}y_{1}\cdots Q_{m}y_{m}\varphi$
with $B\subseteq\mathsf{L}$, we first remove the quantifications
over all fictive variables of $\varphi$ (and eliminate the fictive
variables if necessary). If quantifiers remain, $\phi$ is either
tautological (if the rightmost quantifier is $\exists$) or unsatisfiable
(if the rightmost quantifier is $\forall$), so the problems are trivial,
and $d_{f}(\boldsymbol{a},\boldsymbol{b})=|\boldsymbol{a}-\boldsymbol{b}|$
for any two solutions $\boldsymbol{a}$ and $\boldsymbol{b}$. Otherwise,
we have a quantifier-free formula and the statements follow from \prettyref{lem:L}.

2. Again as in \prettyref{lem:-gopa}, it follows that \noun{st-QBF-Conn(}$B$)
and \noun{QBF-Conn(}$B$) are in $\mathrm{PSPACE}$, since the evaluation
problem for quantified $B$-formulas is in $\mathrm{PSPACE}$ \citep{Schaefer:1978:CSP:800133.804350}.

An inspection of Post's lattice shows that if $B\nsubseteq\mathsf{M}$
and $B\nsubseteq\mathsf{L}$, then $[B]\supseteq\mathsf{S}_{12}$,
$[B]\supseteq\mathsf{D}_{1}$, or $[B]\supseteq\mathsf{S}_{02}$,
so we have to prove $\mathrm{PSPACE}$-completeness and show the existence
of $B$-formulas with an exponential diameter in these cases.

For $[B]\supseteq\mathsf{S}_{12}$ and $[B]\supseteq\mathsf{D}_{1}$,
the statements for the $\mathrm{PSPACE}$-hardness and the diameter
obviously carry over from \prettyref{thm:func}.

For $B\supseteq\mathsf{S}_{02}$, we give a reduction from the problems
for (unquantified) 3-CNF-formulas; we proceeded again similar as in
the proof of \prettyref{lem:s12}. We give a transformation $T_{\psi}$
s.t. $T_{\psi}\in\mathsf{S}_{02}$ for all formulas $\psi$. Since
$\mathsf{S}_{02}=\mathsf{S}_{0}\cap\mathsf{R}_{0}\cap\mathsf{R}_{1}$,
$T_{\psi}$ must be self-dual, 0-reproducing, and 1-reproducing. For
$\psi(\boldsymbol{x})$ let 
\[
T_{\psi}=(\psi\wedge y)\vee z,
\]
with the two new variables $y$ and $z$.

$T_{\psi}$ is 0-separating since all vectors that are no solutions
have $z=0$. Also, $T_{\psi}$ is 0-reproducing and 1-reproducing.
Again, we use the algorithm \noun{Tr} from the proof of \prettyref{lem:s12}
to transform any 3-CNF-formula $\phi$ into a $B$-formula $\varphi'$
equivalent to $T_{\phi}$. Again, we show

\begin{eqnarray*}
T_{\wedge}(T_{\psi_{1}},T_{\psi_{2}}) & = & \left(\left((\psi_{1}\wedge y)\vee z\right)\wedge\left((\psi_{2}\wedge y)\vee z\right)\wedge y\right)\vee z\\
 & \equiv & \left(\left(\psi_{1}\wedge y\right)\wedge\left(\psi_{2}\wedge y\right)\wedge y\right)\vee z\\
 & \equiv & \left(\psi_{1}\wedge\psi_{2}\wedge y\right)\vee z=T_{\psi_{1}\wedge\psi_{2}}.
\end{eqnarray*}

Now let 
\[
\phi'=\forall z\varphi'.
\]
Then, for any two solutions $\boldsymbol{s}$ and $\boldsymbol{t}$
of $\phi(\boldsymbol{x})$, $\boldsymbol{s}'=\boldsymbol{s}\cdot1$
and $\boldsymbol{t}'=\boldsymbol{t}\cdot1$ are solutions of $\phi'(\boldsymbol{x},y)$,
and they are connected in $G(\phi')$ iff $\boldsymbol{s}$ and $\boldsymbol{t}$
are connected in $G(\phi)$, and $G(\phi')$ is connected iff $G(\phi)$
is connected.

The proof of \prettyref{lem:gop} shows that there is an $(n-1)$-ary
function $f$ with diameter of at least $2^{\left\lfloor \frac{n-1}{2}\right\rfloor }$.
Let $f$ be represented by a formula $\phi$; then $\phi'$ as defined
above is a quantified $B$-formula with $n$ free variables and one
quantifier with the same diameter.\end{proof}
\begin{rem}
An analog to \prettyref{thm:quan} also holds for quantified circuits
as defined in \citep[Section 7]{reith2000}.
\end{rem}

\section{Future Directions}

While for $st$-connectivity and connectivity of $B$-formulas and
$B$-circuits we now have a quite complete picture, there is a multitude
of interesting variations in different directions with open problems.

As mentioned in the abstract, for CNF($\mathcal{S}$)-formulas with
constants, we have a complete classification for both connectivity
problems and the diameter also \citep{csp}. However, for CNF($\mathcal{S}$)-formulas
without constants, the complexity of the connectivity problem is still
open in some cases \citep{diss}.

Besides CNF($\mathcal{S}$)-formulas, $B$-formulas and $B$-circuits,
there are further variants of Boolean satisfiability, and investigating
connectivity in these settings might be worthwhile as well. For example,
disjunctive normal forms with special connectivity properties were
studied by Ekin et al. already in 1997 for their ``important role
in problems appearing in various areas including in particular discrete
optimization, machine learning, automated reasoning, etc.'' \citep{ekin1999connected}.

Other connectivity-related problems already mentioned by Gopalan et
al.\ are counting the number of components and approximating the
diameter. Recently, Mouawad et al. investigated the question of finding
the shortest path between two solutions \citep{mouawad2014shortest},
which is of special interest to reconfiguration problems.

Furthermore, our definition of connectivity is not the only sensible
one: One could regard two solutions connected whenever their Hamming
distance is at most $d$, for any fixed $d\geq1$; this was already
considered related to random satisfiability, see \citep{achlioptas2006solution}.
This generalization seems meaningful as well as challenging.

Finally, a most interesting subject are CSPs over larger domains;
in 1993, Feder and Vardi conjectured a dichotomy for the satisfiability
problem over arbitrary finite domains \citep{feder1998computational},
and while the conjecture was proved for domains of size three in 2002
by Bulatov \citep{bulatov2002dichotomy}, it remains open to date
for the general case. Close investigation of the solution space might
lead to valuable insights here.

For $k$-colorability, which is a special case of the general CSP
over a $k$-element set, the connectivity problems and the diameter
were already studied by Bonsma and Cereceda \citep{bonsma2009finding},
and Cereceda, van den Heuvel, and Johnson \citep{cereceda2011finding}.
They showed that for $k=3$ the diameter is at most quadratic in the
number of vertices and the $st$-connectivity problem is in P, while
for $k\geq4$, the diameter can be exponential and $st$-connectivity
is PSPACE-complete in general.\bibliographystyle{amsalpha}
\bibliography{jib}

\end{document}

%% file: lattice.tex
	% BF and c-Reproducing Functions
	\node[clone] (BF) 	at (0,24) 		{\hphantom{$\CloneBF$}};
	\node[clone] (R1) 	at (-1,23) 	{\hphantom{$\CloneR_1$}};
	\node[clone] (R0) 	at (1,23)		{\hphantom{$\CloneR_0$}};
	\node[clone] (R2) 	at (0,22) 		{\hphantom{$\CloneR_2$}};
	
	% Montone Functions
	\node[clone] (M) 	at (0,20) 		{\hphantom{$\CloneM$}};
	\node[clone] (M1) 	at (-1,19) 		{\hphantom{$\CloneM_1$}};
	\node[clone] (M0) 	at (1,19)		{\hphantom{$\CloneM_0$}};
	\node[clone] (M2) 	at (0,18) 		{\hphantom{$\CloneM_2$}};

	% 1-separating 1
	\node[clone] (S21) 	at (6.7,19.3)		 		{\hphantom{$\CloneS_{1}^2$}};
	\node[clone] (S31) 	at (6.7,17.8)		 		{\hphantom{$\CloneS_{1}^3$}};
	\node[clone] (Sn1) 	at (6.7,15.2)		 		{\hphantom{$\CloneS_{1}^n$}};
	\node[clone] (S1)	 	at (6.7,13.8) 				{\hphantom{$\CloneS_{1}$}};

	% 1-separating 2
	\node[clone] (S212) 	at (5.8,18.0)		 		{\hphantom{$\CloneS_{12}^2$}};
	\node[clone] (S312) 	at (5.8,16.5)		 		{\hphantom{$\CloneS_{12}^3$}};
	\node[clone] (Sn12) 	at (5.8,14.0)		 		{\hphantom{$\CloneS_{12}^n$}};
	\node[clone] (S12)	 	at (5.8,12.5) 				{\hphantom{$\CloneS_{12}$}};

	% 1-separating 3
	\node[clone] (S211) 	at (4.8,18.0)		 		{\hphantom{$\CloneS_{11}^2$}};
	\node[clone] (S311) 	at (4.8,16.5)		 		{\hphantom{$\CloneS_{11}^3$}};
	\node[clone] (Sn11) 	at (4.8,14.0)		 		{\hphantom{$\CloneS_{11}^n$}};
	\node[clone] (S11)	 	at (4.8,12.5) 				{\hphantom{$\CloneS_{11}$}};

	% 1-separating 4
	\node[clone] (S210) 	at (3.9,16.7)			{\hphantom{$\CloneS_{10}^2$}};
	\node[clone] (S310) 	at (3.9,15.2) 			{\hphantom{$\CloneS_{10}^3$}};
	\node[clone] (Sn10) 	at (3.9,12.7) 			{\hphantom{$\CloneS_{10}^n$}};
	\node[clone] (S10)	 	at (3.9,11.2) 			{\hphantom{$\CloneS_{10}$}};

	% 0-separating 1
	\node[clone] (S20) 	at (-6.7,19.3)		 		{\hphantom{$\CloneS_{0}^2$}};
	\node[clone] (S30) 	at (-6.7,17.8)		 		{\hphantom{$\CloneS_{0}^3$}};
	\node[clone] (Sn0) 	at (-6.7,15.3)		 		{\hphantom{$\CloneS_{0}^n$}};
	\node[clone] (S0)	 	at (-6.7,13.8) 			{\hphantom{$\CloneS_{0}$}};

	% 0-separating 2
	\node[clone] (S202) 	at (-5.8,18.0)		 		{\hphantom{$\CloneS_{02}^2$}};
	\node[clone] (S302) 	at (-5.8,16.5)		 		{\hphantom{$\CloneS_{02}^3$}};
	\node[clone] (Sn02) 	at (-5.8,14.0)		 		{\hphantom{$\CloneS_{02}^n$}};
	\node[clone] (S02)	 	at (-5.8,12.5) 			{\hphantom{$\CloneS_{02}$}};

	% 0-separating 3
	\node[clone] (S201) 	at (-4.8,18.0)		 		{\hphantom{$\CloneS_{01}^2$}};
	\node[clone] (S301) 	at (-4.8,16.5)		 		{\hphantom{$\CloneS_{01}^3$}};
	\node[clone] (Sn01) 	at (-4.8,14.0)		 		{\hphantom{$\CloneS_{01}^n$}};
	\node[clone] (S01)	 	at (-4.8,12.5) 			{\hphantom{$\CloneS_{01}$}};
	
	% 0-separating 4
	\node[clone] (S200) 	at (-3.9,16.7)		 		{\hphantom{$\CloneS_{00}^2$}};
	\node[clone] (S300) 	at (-3.9,15.2)		 		{\hphantom{$\CloneS_{00}^3$}};
	\node[clone] (Sn00) 	at (-3.9,12.7)		 		{\hphantom{$\CloneS_{00}^n$}};
	\node[clone] (S00)	 	at (-3.9,11.2) 			{\hphantom{$\CloneS_{00}$}};
	
	% Selfdual Functions
	\node[clone] (D) 	at (0,14) 		{\hphantom{$\CloneD$}};
	\node[clone] (D1) 	at (0,13) 		{\hphantom{$\CloneD_1$}};
	\node[clone] (D2) 	at (0,12) 		{\hphantom{$\CloneD_2$}};
	
	% Conjunctions
	\node[clone] (E) 	at (3,10) 		{\hphantom{$\CloneE$}};
	\node[clone] (E1) 	at (2,9) 		{\hphantom{$\CloneE_1$}};
	\node[clone] (E0) 	at (4,9)			{\hphantom{$\CloneE_0$}};
	\node[clone] (E2) 	at (3,8) 		{\hphantom{$\CloneE_2$}};
	
	% Disjunctions
	\node[clone] (V) 	at (-3,10) 	{\hphantom{$\CloneV$}};
	\node[clone] (V0) 	at (-2,9)		{\hphantom{$\CloneV_0$}};
	\node[clone] (V1) 	at (-4,9) 		{\hphantom{$\CloneV_1$}};
	\node[clone] (V2) 	at (-3,8) 		{\hphantom{$\CloneV_2$}};
	
	% Affine Functions
	\node[clone] (L) 	at (0,10) 		{\hphantom{$\CloneL$}};
	\node[clone] (L0) 	at (1,9) 		{\hphantom{$\CloneL_0$}};
	\node[clone] (L1) 	at (-1,9) 		{\hphantom{$\CloneL_1$}};
	\node[clone] (L3) 	at (0,9)		 	{\hphantom{$\CloneL_3$}};
	\node[clone] (L2) 	at (0,8) 		{\hphantom{$\CloneL_2$}};
	
	% Negations
	\node[clone] (N) 	at (0,6.5) {\hphantom{$\CloneN$}};
	\node[clone] (N2) 	at (0,5.5) {\hphantom{$\CloneN_2$}};
	
	% Identities
	\node[clone] (I) 	at (0,2.5) {\hphantom{$\CloneI$}};
	\node[clone] (I0) 	at (1,1.5) {\hphantom{$\CloneI_0$}};
	\node[clone] (I1) 	at (-1,1.5) {\hphantom{$\CloneI_1$}};
	\node[clone] (I2) 	at (0,0.5) {\hphantom{$\CloneI_2$}};

	% Edges R
	\path[edge]
		(R2)		edge node {} (R1)
					edge node {} (R0)
		(R1)		edge node {} (BF)
		(R0)		edge node {} (BF);

	% Edges M
	\path[edge]
		(M2)		edge node {} (M1)
					edge[out=120,in=240,looseness=0.8] node {} (R2)
					edge node {} (M0)
		(M1)		edge node {} (R1)
					edge node {} (M)
		(M0)		edge node {} (R0)
					edge node {} (M)
		(M)		edge[out=50,in=306,looseness=0.8] node {} (BF);

	% Edges from 1-separating 1
	\path[edge]
		(S1)		edge[densely dashed] node {} (S31)
		(Sn1) edge node {} (Sn11)
				 edge node {} (Sn12)
		(S31)	edge node {} (S21)
		(S21)	edge node {} (R0);

	% Edges from 1-separating 2
	\path[edge]
		(S12)	edge node {} (S1)
					edge[densely dashed] node {} (S312)
		(S312)	edge node {} (S31)
					edge node {} (S212)
		(S212)	edge node {} (S21)
					edge node {} (R2);

	% Edges from 1-separating 3
	\path[edge]
		(S11)	edge node {} (S1)
					edge[densely dashed] node {} (S311)
		(S311)	edge node {} (S31)
					edge node {} (S211)
		(S211)	edge node {} (S21)
					edge node {} (M0);
	
	% Edges from 1-separating 4
	\path[edge]
		(S10)	edge node {} (S12)
					edge node {} (S11)
					edge[densely dashed] node {} (S310)
		(Sn10)	edge node {} (Sn11)
					edge node {} (Sn12)
		(S310)	edge node {} (S312)
					edge node {} (S311)
					edge node {} (S210)
		(S210)	edge node {} (S212)
					edge node {} (S211)
					edge node {} (M2);

	% Edges from 0-separating 1
	\path[edge]
		(S0)		edge[densely dashed] node {} (S30)
		(Sn0)	edge node {} (Sn02)
					edge node {} (Sn01)
		(S30)	edge node {} (S20)
		(S20)	edge node {} (R1);

	% Edges from 0-separating 2
	\path[edge]
		(S02)	edge node {} (S0)
					edge[densely dashed] node {} (S302)
		(S302)	edge node {} (S30)
					edge node {} (S202)
		(S202)	edge node {} (S20)
					edge node {} (R2);

	% Edges from 0-separating 3
	\path[edge]
		(S01)	edge node {} (S0)
					edge[densely dashed] node {} (S301)
		(S301)	edge node {} (S30)
					edge node {} (S201)
		(S201)	edge node {} (S20)
					edge node {} (M1);
	
	% Edges from 0-separating 4
	\path[edge]
		(S00)	edge node {} (S02)
					edge node {} (S01)
					edge[densely dashed] node {} (S300)
		(Sn00)	edge node {} (Sn01)
					edge node {} (Sn02)
		(S300)	edge node {} (S302)
					edge node {} (S301)
					edge node {} (S200)
		(S200)	edge node {} (S202)
					edge node {} (S201)
					edge node {} (M2);
	
	% Edges from D
	\path[edge]
		(D2)		edge node {} (S200)
					edge node {} (S210)
					edge node {} (D1)
		(D1)		edge[out=135,in=220,looseness=0.95] node {} (R2)
					edge node {} (D)
		(D)		edge[out=41,in=306,looseness=0.85] node {} (BF);
	
	% Edges from L
	\path[edge]
		(L2)		edge node {} (L1)
					edge node {} (L3)
					edge node {} (L0)
					edge[out=50,in=310,looseness=0.6] node {} (D1)
		(L1)		edge[out=110,in=250,looseness=0.9] node {} (R1)
					edge node {} (L)
		(L3)		edge node {} (L)
					edge[out=130,in=230,looseness=0.8] node {} (D)
		(L0)		edge node {} (L)
					edge[out=70,in=290,looseness=0.9] node {} (R0)
		(L)		edge[out=55,in=306,looseness=0.8] node {} (BF);

	% Edges from E
	\path[edge]
		(E2)		edge node {} (E1)
					edge node {} (S10)
					edge node {} (E0)
		(E1)		edge node {} (M1)
					edge node {} (E)
		(E0)		edge node {} (E)
					edge node {} (S11)
		(E)		edge node {} (M);
	
	% Edges from V
	\path[edge]
		(V2)		edge node {} (V1)
					edge node {} (S00)
					edge node {} (V0)
		(V1)		edge node {} (S01)
					edge node {} (V)
		(V0)		edge node {} (V)
					edge node {} (M0)
		(V)		edge node {} (M);
	
	% Edges from N
	\path[edge]
		(N2)		edge node {} (N)
					edge[out=50,in=310,looseness=0.8] node {} (L3)
		(N)		edge[out=125,in=230,looseness=0.8] node {} (L);
	
	% Edges from I
	\path[edge]
		(I2) 	edge node {} (I1)
					edge node {} (I0)
					edge node {} (E2)
					edge node {} (V2)
					edge[out=115,in=240,looseness=0.8] node {} (N2)
					edge[out=115,in=240,looseness=0.8] node {} (L2)
					edge[out=70,in=290,looseness=0.7] node {} (D2)
		(I1)		edge node {} (V1)
					edge[out=110,in=250,looseness=0.8] node {} (L1)
					edge node {} (E1)
					edge node {} (I)
		(I0)		edge node {} (I)
					edge[out=70,in=290,looseness=0.8] node {} (L0)
					edge node {} (V0)
					edge node {} (E0)
		(I)		edge node {} (V)
					edge[out=62,in=312,looseness=0.8] node {} (N)
					edge node {} (E);